%% file: robust_auction_arxiv.tex
\documentclass[11pt,twoside]{article}

\usepackage{fullpage}

\input{wg_style}

\input{macros}

\ifnum\ShowAuthNotes=0
\newcommand{\authnote}[2]{{ \footnotesize \bf{\color{DarkRed}[#1's note:
{\color{DarkBlue}#2}]}}}
\else
\newcommand{\authnote}[2]{}
\fi

\newtheorem{counterexample}{Counterexample}

\begin{document}

\begin{center}

  {\bf{\LARGE{Robust Learning of Optimal Auctions}}}

\vspace*{.35in}

{\large{
\begin{tabular}{ccc}
Wenshuo Guo$^{\dagger}$, Michael I. Jordan$^{\dagger, \ddag}$, Manolis Zampetakis$^{\dagger}$
\end{tabular}
}}
\vspace*{.2in}

\begin{tabular}{c}
$^\dagger$Department of Electrical Engineering and Computer Sciences,\\ 
$^\ddag$Department of Statistics, 
\\University of California, Berkeley\\
\end{tabular}

\vspace*{.2in}

\today

\vspace*{.2in}

\begin{abstract}
    \input{sec/abstract}

\end{abstract}
\end{center}

\input{sec/intro}

\input{sec/related_work}

\input{sec/prelim}

\input{sec/UB}

\input{sec/sample_complexity}

\input{sec/conclu}
\subsection*{Acknowledgments}

This work was supported in part by the Mathematical Data Science program of the Office of Naval Research under grant number N00014-18-1-2764.

\bibliography{ref}
\bibliographystyle{abbrvnat}

\clearpage
\appendix

\begin{center}

  {\bf{\LARGE{Appendix}}}
\end{center}

\input{sec/proofs/tech_lemmas}
\input{sec/proofs/multibidder.tex}
\input{sec/proofs/PM_UB_MHR_single}
\input{sec/proofs/PM_LB_examples}
\input{sec/proofs/sample_UB_proofs}
\input{sec/proofs/sample_LB_proofs}

\end{document}

%% file: wg_style.tex
\usepackage{lipsum}
\usepackage{wrapfig}
\usepackage{diagbox}
\usepackage{makecell}
\usepackage{booktabs}
\usepackage{tikz,tikz-3dplot}
\usepackage{amsmath}
\usepackage{amsfonts}
\usepackage{amssymb}
\usepackage{amsthm}
\usepackage{url,ifthen}
\usepackage[square,sort,comma,numbers]{natbib}
\usepackage{enumitem}
\usepackage{srcltx}
\usepackage{multirow}
\usepackage{boxedminipage}
\usepackage{nicefrac}
\usepackage{xspace}
\usepackage{graphicx}
\usepackage{color}
\usepackage{subfig}
\usepackage{colortbl}
\usepackage{setspace}
\usepackage{pgfplots}
\pgfplotsset{compat=1.12}
\usepackage{algorithm,algorithmic}
\usepackage[algo2e]{algorithm2e}

\usepackage{xcolor}
\definecolor{DarkGreen}{rgb}{0.1,0.5,0.1}
\definecolor{DarkRed}{rgb}{0.5,0.1,0.1}
\definecolor{DarkBlue}{rgb}{0.1,0.1,0.5}
\definecolor{Gray}{rgb}{0.2,0.2,0.2}

\definecolor{c1}{RGB}{38, 70, 83}
\definecolor{c2}{RGB}{42, 157, 143}
\definecolor{c3}{RGB}{233, 196, 106}
\definecolor{c5}{RGB}{231, 111, 81}
\definecolor{c4}{RGB}{244, 162, 97}

\definecolor{c1}{RGB}{38, 70, 83}
\definecolor{c2}{RGB}{42, 157, 143}
\definecolor{c3}{RGB}{233, 196, 106}
\definecolor{c5}{RGB}{231, 111, 81}
\definecolor{c4}{RGB}{244, 162, 97}

\usepackage{listings}
\lstdefinestyle{mystyle}{
    commentstyle=\color{DarkBlue},
    keywordstyle=\color{DarkRed},
    numberstyle=\tiny\color{Gray},
    stringstyle=\color{DarkGreen},
    basicstyle=\footnotesize,
    breakatwhitespace=false,         
    breaklines=true,                 
    captionpos=b,                    
    keepspaces=true,                 
    numbers=left,                    
    numbersep=5pt,                  
    showspaces=false,                
    showstringspaces=false,
    showtabs=false,                  
    tabsize=2
}
\usepackage{caption}
\usepackage{hyperref}
\hypersetup{
    unicode=false,          
    pdftoolbar=true,        
    pdfmenubar=true,        
    pdffitwindow=false,      
    pdfnewwindow=true,      
    colorlinks=true,       
    linkcolor=DarkBlue,          
    citecolor=DarkGreen,        
    filecolor=DarkRed,      
    urlcolor=DarkBlue,          
    pdftitle={},
    pdfauthor={},
}

\def\draft{1}

\def\submit{0}

\ifnum\draft=1 
    \def\ShowAuthNotes{1}
\else
    \def\ShowAuthNotes{0}
\fi

\ifnum\submit=1
\newcommand{\forsubmit}[1]{#1}
\newcommand{\forreals}[1]{}
\else
\newcommand{\forreals}[1]{#1}
\newcommand{\forsubmit}[1]{}
\fi
\newtheorem{theorem}{Theorem}[section]

\newtheorem{lemma}[theorem]{Lemma}

\newtheorem{proposition}[theorem]{Proposition}

\theoremstyle{definition}
\newtheorem{definition}[theorem]{Definition}

\newcommand{\chapterref}[1]{\hyperref[ch:#1]{Chapter~\ref{ch:#1}}}

\newcommand{\claimref}[1]{\hyperref[claim:#1]{Claim~\ref{claim:#1}}}

\newcommand{\corollaryref}[1]{\hyperref[cor:#1]{Corollary~\ref{cor:#1}}}

\newcommand{\definitionref}[1]{\hyperref[def:#1]{Definition~\ref{def:#1}}}

\newcommand{\equationref}[1]{\hyperref[eq:#1]{Equation~\ref{eq:#1}}}

\newcommand{\factref}[1]{\hyperref[fact:#1]{Fact~\ref{fact:#1}}}

\newcommand{\figureref}[1]{\hyperref[fig:#1]{Figure~\ref{fig:#1}}}

\newcommand{\tableref}[1]{\hyperref[tab:#1]{Table~\ref{tab:#1}}}

\newcommand{\itemref}[1]{\hyperref[item:#1]{Item~(\ref{item:#1})}}

\newcommand{\lemmaref}[1]{\hyperref[lem:#1]{Lemma~\ref{lem:#1}}}

\newcommand{\propref}[1]{\hyperref[prop:#1]{Proposition~\ref{prop:#1}}}

\newcommand{\propositionref}[1]{\hyperref[prop:#1]{Proposition~\ref{prop:#1}}}

\newcommand{\remarkref}[1]{\hyperref[rem:#1]{Remark~\ref{rem:#1}}}

\newcommand{\sectionref}[1]{\hyperref[sec:#1]{Section~\ref{sec:#1}}}

\newcommand{\theoremref}[1]{\hyperref[thm:#1]{Theorem~\ref{thm:#1}}}

\newcommand{\Esymb}{\mathbb{E}}
\newcommand{\Psymb}{\mathbb{P}}

\DeclareMathOperator*{\E}{\Esymb}

\DeclareMathOperator*{\ProbOp}{\Psymb r}
\renewcommand{\Pr}{\ProbOp}

\renewcommand{\hat}{\widehat}

\newcommand{\cD}{{\cal D}}

\newcommand{\cH}{{\cal H}}

\newcommand{\cX}{{\cal X}}

\newcommand{\defeq}{\stackrel{\small \mathrm{def}}{=}}
\renewcommand{\leq}{\leqslant}
\renewcommand{\le}{\leqslant}
\renewcommand{\geq}{\geqslant}
\renewcommand{\ge}{\geqslant}

\newcommand{\Abs}[1]{\left\lvert#1\right\rvert}

\newcommand{\R}{\mathbb{R}}

\newcommand{\D}{\mathcal D}

\usepackage{bm}

\newcommand{\ignore}[1]{}
\newcommand{\OPT}{\mathrm{OPT}}

\DeclareMathOperator*{\argmin}{arg\,min}

\renewcommand{\epsilon}{\varepsilon}

\newcommand{\eps}{\epsilon}

\newcommand{\remove}[1]{}

\def\mydatewithyear{\leavevmode\hbox{\the\year/\twodigits\month/\twodigits\day}}
\def\twodigits#1{\ifnum#1<10 0\fi\the#1}

\def\twodigits#1{\ifnum#1<10 0\fi\the#1}

\makeatletter
\newtheorem*{rep@theorem}{\rep@title}
\newcommand{\newreptheorem}[2]{%
\newenvironment{rep#1}[1]{%
 \def\rep@title{#2 \ref{##1}}%
 \begin{rep@theorem}}%
 {\end{rep@theorem}}}
\makeatother

\newreptheorem{theorem}{Theorem}
\newreptheorem{corollary}{Corollary}
\newreptheorem{lemma}{Lemma}

%% file: macros.tex
\newcounter{protocol}
\makeatletter
\newenvironment{protocol}[1][htb]{%
  \let\c@algorithm\c@protocol
  \renewcommand{\ALG@name}{Procedure}
  \begin{algorithm}[#1]%
  }{\end{algorithm}
}
\makeatother

\newcommand{\Dtilde}{\tilde{\mathcal{D}}}
\newcommand{\Dstar}{\mathcal{D}^\ast}
\newcommand{\vx}{\mathbf x}

\newcommand{\vE}{\mathbf E}

\newcommand{\Rplus}{\mathbb{R}_{+}}
\newcommand{\vv}{\mathbf v}
\newcommand{\vb}{\mathbf b}
\newcommand{\vp}{\mathbf p}
\newcommand{\vD}{\mathcal{\mathbf{D}}}
\newcommand{\vDstar}{\mathcal{\mathbf{D}}^\ast}
\newcommand{\vDtilde}{\mathcal{\mathbf{\tilde{D}}}}
\newcommand{\REV}{\mathrm{Rev}}

%% file: sec/abstract.tex
We study the problem of learning revenue-optimal multi-bidder auctions from samples when the samples of bidders' valuations can be adversarially corrupted or drawn from distributions that are adversarially perturbed. First, we prove tight upper bounds on the revenue we can obtain with a corrupted distribution under a population model, for both regular valuation distributions and distributions with monotone hazard rate (MHR). We then propose new algorithms that, given only an ``approximate distribution'' for the bidder's valuation, can learn a mechanism whose revenue is nearly optimal simultaneously for all ``true distributions'' that are $\alpha$-close to the original distribution in Kolmogorov-Smirnov distance. The proposed algorithms operate beyond the setting of bounded distributions that have been studied in prior works, and are guaranteed to obtain a fraction  $1-O(\alpha)$ of the optimal revenue under the true distribution when the distributions are MHR.  Moreover, they are guaranteed to yield at least a fraction $1-O(\sqrt{\alpha})$ of the optimal revenue when the distributions are regular. We prove that these upper bounds cannot be further improved, by providing matching lower bounds. Lastly, we derive sample complexity upper bounds for learning a near-optimal auction for both MHR and regular distributions.

%% file: sec/intro.tex
\section{Introduction}

Optimal auctions play a crucial role in economic theory, with a wide range of applications across various industries, public sectors, and online platforms \citep{myerson1981optimal,bykowsky2000mutually, roth2002last, klemperer2002really, milgrom2004putting,  lahaie2007sponsored}. In such auctions, pricing mechanisms need to be determined by the auction designer so as to satisfy various desired goals, such as revenue maximization and incentive compatibility. Often this determination is made based on information about the buyers that is assumed to be available a priori. For example, in a standard valuation model, each bidder has a valuation over the available items, and if the seller knows the distribution of these valuations, she could design an optimal auction which maximizes her revenue. 

Arguably the fundamental difficulty in the design of optimal auctions is that real valuations are private and unknown to the auction designer. Consider specifically the problem of selling one item to multiple buyers. Suppose that we model the buyers' valuations as arising as independent draws from buyer-specific prior distributions. In this scenario, what is the optimal mechanism in terms of the expected revenue? This problem was solved by \citet{myerson1981optimal} through a characterization of \emph{virtual value functions}. In particular, we can define a virtual value function of each buyer based on their prior distributions. An optimal auction then lets the buyer with the largest non-negative virtual value win the item, and charges the winner a price that equals the threshold value above which she wins.\footnote{More generally, the optimal auction picks the winner based on the virtual value after an ``ironing'' procedure.}

Unfortunately, there is a further fundamental challenge in deploying these theoretical results in practice, which is that in real-world settings the auction designer may not even know the prior distributions on valuations. Instead, what the designer might hope for is that there is a stream of previous transactions, or some other relevant auxiliary data, that is helpful in inferring the buyers' private distributions. This perspective has motivated an active recent literature learning optimal auctions from samples~\citep{cole2014sample, devanur2016sample, morgenstern2015pseudo, morgenstern2016learning, syrgkanis2017sample, dudik2017oracle, gonczarowski2017efficient, huang2018making, roughgarden2016ironing, balcan2018general, guo2019settling, roughgarden2019minimizing, gonczarowski2021sample}. In this line of work, the central question is: suppose we are only able to access the prior distributions in the form of independent samples, how many samples are sufficient and necessary for finding an approximately optimal auction?

While this merging of mechanism design and learning theory is appealing, a further concern arises.  Given the potentially adversarial setting of auction design, do we really believe that the data that we observe are drawn in accord with our assumptions?  More concretely, is the learning of optimal auctions robust to adversarial corruptions of the samples? This problem is arguably at the core of what it means to learn an optimal auction.  It is a challenging problem; indeed, as we show in Counterexample \ref{cex:necessity} in Section \ref{sec:samples}, auction designs that are optimal in the absence of corruptions can become arbitrarily bad even if a small portion of the samples are corrupted.  Building on earlier work by \citet{cai2017learning} and \citet{brustle2020multi}, we tackle a key open problem---what is the best approximation to the optimal revenue for arbitrary levels of corruption for distributions with unbounded support? And what is the mechanism that achieves it?

In summary, in this work we explore the problem of the robust learning of optimal auctions, where the samples of bidders' valuations are subject to corruption and their support is unbounded. In particular, we consider having access to samples that are drawn from some distribution $\Dtilde$ which is within a Kolmogorov-Smirnov (KS) distance $\alpha$ of the true distribution $\Dstar$. Denote $\OPT$ as the maximum revenue we can achieve under the true valuation distributions. Our goal is to design mechanisms that are guaranteed to achieve a revenue of at least $(1-\rho(\alpha)) \cdot \OPT$ for the smallest possible error $\rho(\alpha)$ and with the use of a minimal number of samples. 

\subsection{Our results} 

We study the problem of learning revenue-optimal multi-bidder auctions from samples when the samples of bidders' valuations can be adversarially corrupted or drawn from distributions that are adversarially perturbed. We summarize our main results as follows:
\begin{enumerate}
    \item We derive tight upper bounds on the revenue we can obtain with a corrupted distribution under a population model. For distributions with monotone hazard rate (MHR), and with total corruption $\alpha$, we obtain an approximation ratio of $1 - O(\alpha)$ compared to the optimal revenue under the true distribution (see Theorem~\ref{thm:mhr-inf-sample}). For regular valuation distributions, where for total corruption $\alpha$, we get an approximation ratio of $1 - O(\sqrt{\alpha})$ (see Theorem~\ref{thm:regular-inf-sample}).
    
    \item To achieve these upper bounds, we propose a new \emph{theoretical} algorithm for the population model (see Algorithm~\ref{Alg:inf-sample-single-bidder}) that, given only an ``approximate distribution'' for the bidder's valuation, can learn a mechanism whose revenue is nearly optimal simultaneously for all ``true distributions'' that are $\alpha$-close to the given distribution in Kolmogorov-Smirnov distance. The proposed algorithm operates beyond the setting of bounded distributions that have been studied in prior works; indeed, they apply to general unbounded MHR and regular distributions.
    
    \item We further show that these upper bounds under the population model cannot be further improved (up to constant log factors), by providing matching lower bounds for both the MHR and regular distributions (see Theorem~\ref{thm:mhr-inf-sample-LB} and Theorem~\ref{thm:regular-inf-sample-LB}).
    
    \item Lastly, we derive sample complexity upper bounds for learning a near-optimal auction for both MHR and regular distributions with multiple bidders (Theorem~\ref{thm:regular-sample-UB} and Theorem~\ref{thm:mhr-sample-UB}), and propose a \emph{practical} algorithm (see Algorithm~\ref{alg:emp-single-bidder}) which takes samples as input. We also provide accompanying sample complexity lower bounds (Theorem~\ref{thm:sample-LB}), and demonstrate a small gap relative to the corresponding upper bounds.
\end{enumerate}

%% file: sec/related_work.tex
\subsection{Related work} \label{sec:related_work}

Designing revenue optimal auctions is a classic problem in economic theory that has attracted much research attention. We survey the most closely related work  in two main areas.

\paragraph{Learning optimal auctions from samples.}

Recent work has explored settings of learning approximately optimal auction from samples, both for  single-item auctions~\citep{cole2014sample}, and multi-item auctions~\citep{balcan2018general, balcan2016sample, morgenstern2015pseudo, syrgkanis2017sample}. Most recently,~\citet{guo2019settling} provide a complete set of sample complexity bounds for single-item auctions, by deriving matching upper and lower bounds up to a poly-logarithmic factor. While these approaches have obtained fruitful results on the sample complexity of learning optimal auctions, a key assumption that is commonly made in this work is that the samples are independently and identically drawn from the bidders' valuation distributions, with the goal of learning an auction which maximizes the expected revenue on the underlying, unknown distribution over bidder valuations. A major difference in our work is that we consider that the samples can suffer from potential corruptions, which is a significantly more challenging setting.

\paragraph{Robustness of learning optimal auctions.}
Our paradigm on the robust learning of optimal auctions is closely related to recent work that considers the learning of auctions from mismatched distributions or corrupted samples. \citet{cai2017learning} consider a multi-item auction setting, where there is a given ``approximate distribution,'' and the goal is to compute an auction whose revenue is approximately optimal simultaneously for all ``true distributions" that are close to the given one. They provide an algorithm that achieves a poly-$\alpha$ additive loss compared to the true optimal revenue. More recently, \citet{brustle2020multi} consider learning multi-item auctions where bidders’ valuations are drawn from correlated distributions that can be captured by Markov random fields. However, they make a key simplifying assumption---that the bidders’ valuation for the items lie in some bounded interval. Our results, by contrast, apply to the general setting of unbounded valuation distributions, a setting that requires new theoretical machinery. To the best of our knowledge, our work constitutes the first analysis of the learnability of single-item optimal auctions from corrupted samples for unbounded distributions.

\paragraph{Organization.}
 In Section~\ref{sec:prelim},  we provide background on auction models and formally state our problem. Section~\ref{sec:UB-inf} contains our main theoretical statements for the population model. We propose an algorithm that achieves optimal theoretical upper bounds, by providing matching lower bounds. Section~\ref{sec:samples} contains our main results on learning with finite samples. We provide a practical algorithm that takes samples from the corrupted distribution, and provides sample complexity upper and lower bounds for both the regular and MHR distributions cases. We conclude in Section~\ref{sec:conclu}.

%% file: sec/prelim.tex
\section{Preliminaries}\label{sec:prelim}

We begin by formally defining the setting we study for robust learning of optimal auctions, which includes the revenue objective and the general classes of valuation distributions that we consider.

\subsection{Auction models}

\noindent \textbf{Single-bidder setting.}\vspace{1mm}
Consider one item for sale to one bidder. The bidder has a private
valuation $v \in \Rplus$ for this item. We assume that $v$ is a random variable
distributed according to the distribution $\Dstar$, with support $\Rplus$, cumulative 
distribution function $F$, and probability density function $f$. 

It is well known that the optimal auction in this setting is a reserve price auction, such that 
 the fask for the seller is to compute a reserve price
$p$ that optimizes revenue \citep{myerson1981optimal}. We assume that the bidder has a quasi-linear utility that is
equal to $u(p) = v-p$ if she decides to buy the item and $u(p) = 0$ otherwise. The seller
aims to set $p$ such that her expected revenue---i.e., the received payment---is maximized. 
We consider the setting where both $v$ and $\Dstar$ are unknown to the seller. However, 
the seller can access i.i.d.\ samples that are drawn from a distribution $\Dtilde$, which is
$\alpha$-close to $\D$ with regard to the Kolmogorov distance: 
\begin{definition}(Kolmogorov-Smirnov distance)
For probability measures $\mu$ and $\nu$ on $\R$, define
\[
d_k(\mu, \nu)  = \sup_{x \in \R} |\mu((-\infty, x)) - \nu((-\infty, x))|.
\]
\end{definition}
It is well known that $d_k(\mu, \nu) \le d_{TV}(\mu, \nu)$, where $d_{TV}$ denotes the 
total variation (TV) distance between $\mu$ and $\mu$. The closeness of $\Dtilde$ to $\Dstar$ is thus formalized as follows:
\begin{align*}
    d_k(\Dstar, \Dtilde) \leq \alpha,
\end{align*}
for some $\alpha > 0$.

\noindent \textbf{Multi-bidder setting.}
Consider one item for sale to $n$ bidders. Each bidder has a private valuation,
$v_i \in \Rplus$, where $v_i$ is independently drawn from the corresponding prior
distribution $\Dstar_i$. Thus, the valuations $\vv =(v_1, v_2, \cdots, v_n)$ follow a
product distribution $\vDstar = \Dstar_1 \times \cdots \times \Dstar_n$. Each bidder
submits a bid $b_i \geq 0$. Denote all the bids as $\vb = (b_1,\cdots, b_n)$. A mechanism
in this setting consists of two rules: the allocation rule $\vx(\vb)$ that takes the bids
$\vb$ and outputs the probability $x_i(\vb)$ that each bidder $i$ will receive the item,
and the payment rule $\vp(\vb)$ that takes the bids $\vb$ and outputs the payment of
bidder $i$. Bidder $i$’s utility is then $u_i(\vb) = v_i \cdot x_i(\vb)-p_i(\vb)$.
The goal of the seller is to find a mechanism that maximizes the expected revenue 
$\E[\sum_{i \in [n]} p_i(\vb)]$, where the expectation is over $\vv \sim \vDstar$, under
the following \textit{Dominant Strategy Incentive Compatibility (DSIC)} and the 
\textit{Individual Rationality (IR)} constraints:
\begin{align}
  u_i(v_i, \vb_{-i}) & \ge u_i(b_i, \vb_{-i}) \quad &\text{for all $v_i, b_i \in \R_+$ and all $\vb_{-i} \in \R_+^{n - 1}$} \tag{DSIC} \label{eq:dsic} \\
  u_i(v_i, \vb_{-i}) & \ge 0 \quad &\text{for all $v_i \in \R_+$ and all $\vb_{-i} \in \R_+^{n - 1}$.} \tag{IR} \label{eq:ir}
\end{align}
We consider the setting in which the valuations and the prior distributions are unknown to 
the seller. Instead, the seller has access to a finite number of i.i.d.\ samples
drawn from the product distribution $\Dtilde = \Dtilde_1 \times \cdots \times \Dtilde_n$, 
where each $\Dtilde_i$ satisfies
\begin{align*}
    d_k(\Dstar_i, \Dtilde_i) \leq \alpha_i,
\end{align*}
for some $\alpha_i > 0, \forall i \in [n]$.

\noindent \textbf{Revenue objective.} Letting $\vD$, $\vD'$ be product or single
bidder distributions as described above, we define $M_{\vD}$ to be the mechanism 
that achieves the optimal revenue for the value distributions $\vD$ and
$\OPT(\vD)$ its expected revenue. Let also $\REV(M_{\vD}, \vD')$ be the expected
revenue of the mechanism $M_{\vD}$ when applied to a setting where the values are
drawn with respect to $\vD'$. 

\subsection{Monotone hazard rate (MHR) and regular distributions}

For any bidder $i$ with a valuation $v_i \sim \D_i$, define the \emph{virtual value function} for this bidder as $\phi_i(v) \defeq v-\frac{1-F_i(v)}{f_i(v)}$, where $F_i$ and $f_i$ are the CDF and PDF of $\D_i$. The \emph{hazard rate} of the distribution $\D_i$ is defined as the function $ \frac{f_i(v)}{1-F_i(v)}$. Then, the distribution $\D_i$ is said to be \textit{regular} if the virtual value $\phi_i(v)$ is monotonically non-decreasing in $v$. Further, distribution $\D_i$ has \textit{monotone hazard rate} (MHR) if $\frac{f_i(v)}{1-F_i(v)}$ is monotone non-decreasing.

%% file: sec/UB.tex
\section{The Population Model} \label{sec:UB-inf}

In this section, we study the problem of learning optimal auction
assuming that we have the exact knowledge of the adversarially perturbed
distributions $\vDtilde$. We relax this assumption in Section
\ref{sec:samples} where we show how to learn optimal auctions when we only have
sample access to $\vDtilde$.

We begin in Section \ref{sec:population:algo} with
the description of our mechanism in the population model. Then, in Section 
\ref{sec:population:mhr}, we present our analysis for the population mechanism for 
Monotone Hazard Rate distributions and we also present the sketch of our proof 
for the single-bidder case. Similarly, in Section \ref{sec:population:regular} 
we state our analysis for the population mechanism for regular distributions and
we present a proof sketch for the single-bidder case. Finally, we show that our proposed mechanism achieves 
optimal (up to constants) guarantees among any mechanism in the population model.

\subsection{Robust Myerson auction in the population model} \label{sec:population:algo}

  Our algorithm assumes as an input the exact knowledge of a product distribution,
$\vDtilde = \Dtilde_1 \times \cdots \times \Dtilde_n$, such that the 
$d_k(\Dstar_i, \Dtilde_i) \le \alpha_i$ and its goal is to find a mechanism that 
achieves approximately optimal revenue for $\vDstar$, where $\vDstar = \Pi_i \Dstar_i$. Without further assumptions, this is an impossible task, as we explain in Section~\ref{sec:samples} via an example.
Thus we assume that the 
algorithm possesses some additional knowledge regarding $\Dstar_i$, either that it is MHR or regular,
and the mechanism needs to exploit this additional property.

To utilize the additional property of the distributions $\Dstar_i$, our mechanism
uses the important concept of the \textit{link function} for MHR and regular distributions. 

\begin{definition}[Link Function]
  The link function $h_M(x; F)$ for MHR distributions is defined as 
  $h_M(x; F) = - \ln(1 - F(x))$ and the link function $h_r(x; F)$ for regular distributions is defined as $h_r(x; F) = 1/(1 - F(x))$. We also define the 
  corresponding inverse link functions $h^{-1}_M(x; h) = 1 - \exp(- h(x))$ and
  $h^{-1}_r(x; h) = 1 - 1/h(x)$. Observe that $h^{-1}_M(x; h_M(x; F)) = F(x)$
  and $h^{-1}_r(x; h_r(x; F)) = F(x)$. We may write $h_M(x)$ or $h_r(x)$ when
  $F$ is clear from the context.
\end{definition}

We provide some intuition on the link function. First, by construction, the link function of either an MHR distribution or a regular distribution is convex and non-decreasing. Second, the link function is monotone with regard to $F$. These two properties are important when we define the notion of a minimal MHR/regular distribution in a Kolmogorov ball, momentarily, which will be used as a necessary step in our algorithm.

Importantly, the link function provides a convenient characterization of the optimal reserve price and optimal revenue for a distribution $F$ that is MHR or regular. To see this, first consider a single bidder with a valuation distribution $F$. Denote the optimal reserve price for selling one item to her as $x^\ast$, and the optimal expected revenue as $\OPT(F)$. Then, when $F$ is MHR, we show that $x^\ast$ is also the unique minimizer of $(h_M(x) - \log(x))$. On the other hand, when $F$ is regular, $v^\ast$ is the point where $h_r(x)$ intersects with its tangent line $kx$, with $k = \nicefrac{1}{\OPT(F)}$ (proof details in Appendix). Figure~\ref{fig:link-function-plot} illustrates such a useful property for $h_M$ and $h_r$ explicitly, for a single-item, single-bidder auction.

\begin{figure*}[!ht]
\centering
\begin{tabular}{cc} 
\includegraphics[width=0.35\textwidth, height=0.33\textwidth]{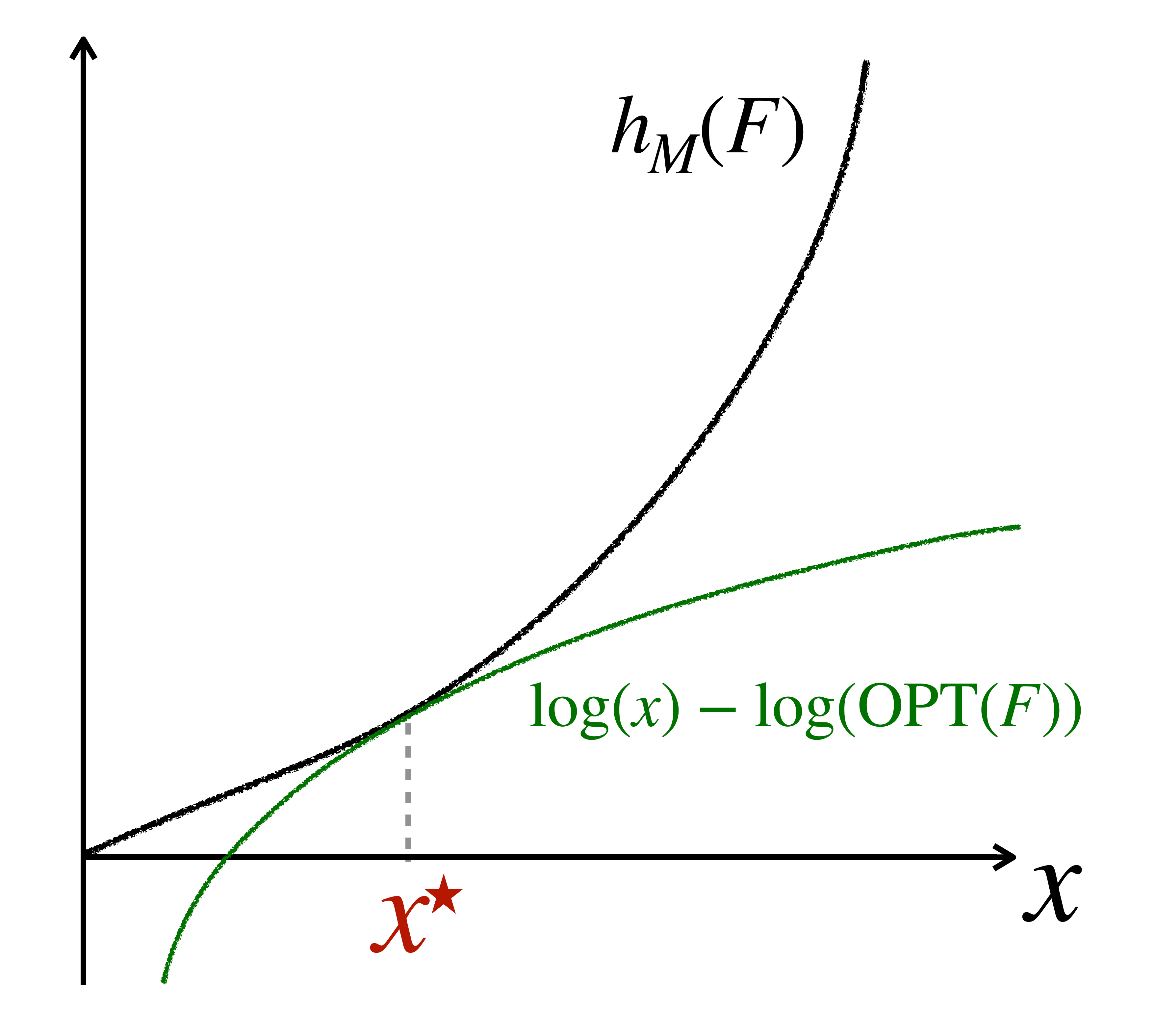} &
\includegraphics[width=0.39\textwidth]{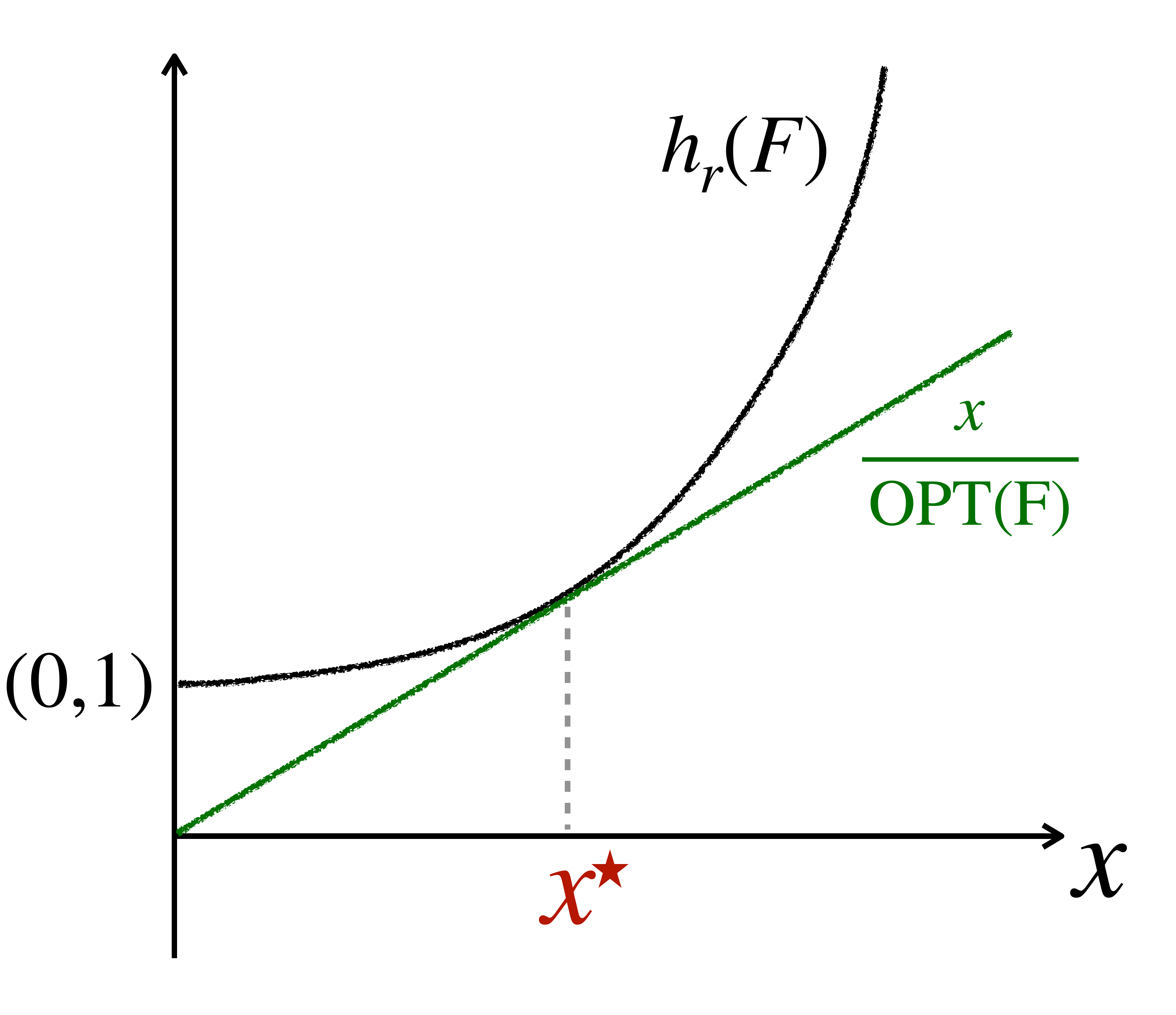} 
\end{tabular}
\caption{Optimal reserve price $x^\ast$ with regard to the link function, for a single-item single-bidder auction with a valuation distribution $F$. \textit{(left)} $F$ is MHR; \textit{(right)} $F$ is regular. }\label{fig:link-function-plot} 
\end{figure*}

Next, we formally define stochastic dominance between two distributions, and state the property of strong revenue monotonicity.

\begin{definition}[Stochastic dominance]\label{def:stochastic-dominance} Given two distributions $\cD_1$ and $\cD_2$ with CDFs as $F_1$ and $F_2$. Then, we say $\cD_1$ (first-order) stochastically dominates $\cD_2$ if for every $x \in \cX$,
\[
F_1(x) \le F_2(x),
\]
denoted as $\cD_1 \succeq \cD_2$. We say a product distribution $\vD = \Pi_i \cD_i$ (component-wise) stochastically dominates another product distribution $\vD' = \Pi_i \cD_i'$ if for every $i$, we have $\cD_i \succeq \cD_i'$.
\end{definition}

\begin{lemma}[Strong revenue monotonicity~\citep{guo2019settling}]\label{lem:strong-rev-monotone}
Let $\vD$, $\vD'$ be two product distributions such that $\vD' \succeq \vD$, then, for $M$ that is the optimal mechanism for $\vD$, we have:
\[
\REV(M, \vD) \leq \REV(M, \vD').
\]

\end{lemma}

The following lemma illustrates the importance of the link functions as well as 
their connection with first-order stochastic dominance. The proof of this lemma is
given in Appendix~\ref{app:sec:tech-lemmas}.

\begin{lemma} \label{lem:linkFunctions}
    A distribution with CDF $F$ is MHR if and only if $h_M(x; F)$ is a convex 
  function of $x$. Similarly, $F$ is regular if and only if $h_r(x; F)$ is a 
  convex function of $x$. 
  Moreover, for two MHR (resp. regular) distributions $F_1$ and $F_2$, such that $F_1 \succeq F_2$, we have that $h_M(x; F_1) \le h_M(x; F_2)$ (resp. $h_r(x; F_1) \le h_r(x; F_2)$) for all $x$.
\end{lemma}

\begin{figure}
\begin{minipage}{0.6\textwidth}
\begin{algorithm}[H]
\begin{algorithmic}[1]
\caption{Robust Myerson Auction in the Population Model}\label{Alg:inf-sample-single-bidder}
\STATE \textbf{Input:} $\alpha_1 \ldots \alpha_n > 0$, link function $h(\cdot)$, possibly corrupted valuation distribution $\tilde F = \Pi_{i=1}^n \tilde F_i$.
\FOR{i = 1\ldots n}
    \STATE Compute a minimal regular / MHR distribution in $B_{d_k, \alpha_i}(\tilde F_i)$ according to Eq~\eqref{eq:minimal-F}, denote as $\hat F_i$.
\ENDFOR
\STATE{Set $\hat F = \Pi_{i=1}^n \hat F_i.$}
\STATE Output Myerson’s optimal auction $M_{\hat{F}}$ w.r.t. the distribution $\hat F$.
\end{algorithmic}
\end{algorithm}
\end{minipage}
\hfill
\begin{minipage}{0.36\textwidth}
    \centering
     \includegraphics[width=0.8\textwidth]{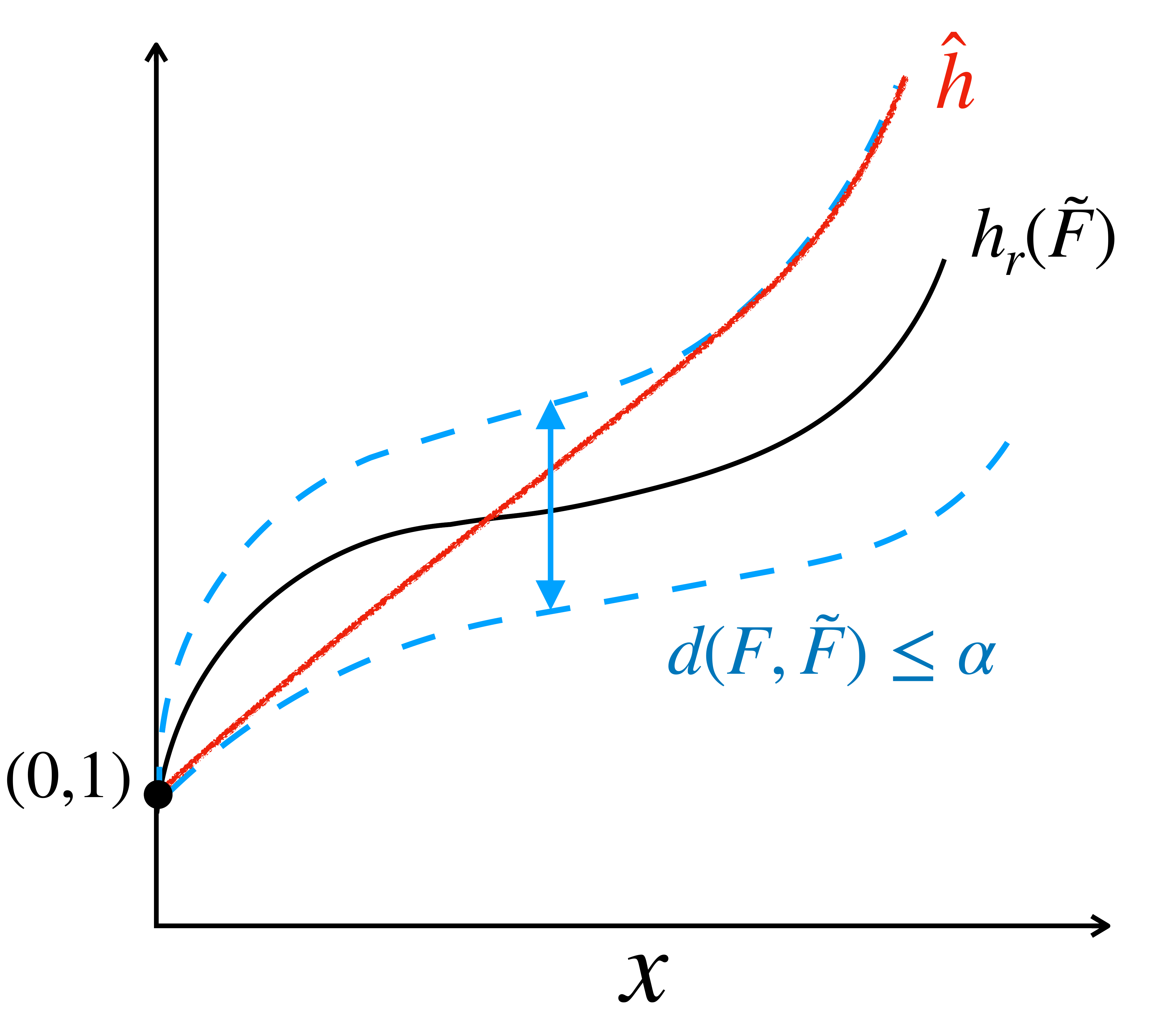}
    \caption{A minimal regular distribution in $B_{d_k,\alpha}$, in the space transformed by applying the link function.}
    \label{fig:minimal-regular}
\end{minipage}
\end{figure}

A key idea used in our algorithm is the minimal MHR/regular distribution within a Kolmogorov distance divergence ball. Formally, 
\begin{definition}\label{def:ks-ball-minimal-F}
For a given distribution $F$, denote the set of all the distributions that are $\alpha$-close to $F$ in Kolmogorov distance as $B_{d_k, \alpha}(F)$:
\[
B_{d_k, \alpha}(F) \defeq \{F': d_k(F', F) \leq F\}.
\]
Further, define a minimal MHR/regular 
distribution within $B_{d_k, \alpha}(F)$ as:
\begin{equation}\label{eq:minimal-F}
\hat F(x) = h^{-1}(x;\hat h), ~~~\text{where}~~~~~
 \hat h(x) \defeq \max_{\substack{\bar F \in B_{d_k, \alpha}(\tilde F)\\ F \text{ is MHR / regular}}} h\left(\bar F(x)\right) ~~~ \forall x \in \R_+.
\end{equation}
\end{definition}

Figure~\ref{fig:minimal-regular} gives an illustration of a minimal regular distribution within $B_{d_k, \alpha}(F)$, in the space transformed by the link function of regular distributions.

\subsection{Analysis for MHR distributions} \label{sec:population:mhr}

  In this section we state the results for the performance of Algorithm 
\ref{Alg:inf-sample-single-bidder} for MHR distributions and we provide a
proof sketch for the single-bidder case. The full proof of the following theorem
can be found in Appendix \ref{sec:multibidderUpperBound}.

\begin{theorem}\label{thm:mhr-inf-sample}
     Let $\vDstar = \Dstar_1 \times \cdots \times \Dstar_n$ be a product 
  distribution where every $\Dstar_i$ is MHR. Let also 
  $\vDtilde = \Dtilde_1 \times \cdots \Dtilde_n$ be any product distribution 
  such that for all $i \in [n]$ it holds that 
  $d_k(\Dstar_i, \Dtilde_i) \le \alpha_i$. If $\tilde{M}$ is the mechanism
  that Algorithm~\ref{Alg:inf-sample-single-bidder} outputs with input $\vDtilde$ then 
  it holds that
  \[ \REV(\tilde{M}, \vDstar) \ge \left(1 - \tilde{O} \left( \sum_{i = 1}^n \alpha_i \right) \right) \cdot \OPT(\vDstar). \]
  In particular for $n = 1$, if $\alpha = \alpha_1$, then we have that
  $\REV(\tilde{M}, \Dstar) \ge \left(1 - O \left( \alpha \right) \right) \cdot \OPT(\Dstar)$.
\end{theorem}

\begin{proof}[Proof sketch for $n = 1$.]
The first key step in our proof is the observation that, by construction, Algorithm~\ref{Alg:inf-sample-single-bidder} runs the Myerson optimal auction on an MHR distribution $\hat F$, such that $\hat F$ is stochastically dominated by any other MHR distribution that is within $B_{d_k, \alpha}(\tilde F')$. On the other hand we have $d_k(F^\ast(x), \tilde F(x)) \leq \alpha$. Applying the triangle inequality, we have $d_k(F^\ast(x), \hat F(x)) \leq 2\alpha$. It is then sufficient for us to bound the ratio of the optimal revenue for any two MHR distributions $F_1$ and $F_2$, with $d_k(F_1, F_2) \leq 2\alpha$, and where $F_1$ is stochastically dominated by $F_2$.

The key part of our proof then considers such $F_1$, $F_2$, and due to the fact that the ratio of the revenues, $\OPT_{F_1}/\OPT_{F_2}$, is scale invariant, we assume without loss of generality that $\OPT_{F_1} = 1$. We then prove that this leads to $h(P_{F_1}^\ast) \leq 1$. The result then follows from two further key lemmas. First, for any reserve price $x < P_{F_1}^\ast$, $ |h_1(x) - h_2(x)|= \left|\log \left(\frac{1-F_2(x)}{1-F_1(x)}\right)\right|$. Further applying the fact that by assumption $|F_1(x) - F_2(x)| \leq \alpha$ we show that $ |h_1(x) - h_2(x)| = O(\alpha)$ for any reserve price $x < P_{F_1}^\ast$. Second, using the fact that $F_1$ is stochastically dominated by $F_2$, we derive that $P_{F_2}^\ast \leq P_{F_1}^\ast$. The conclusion then follows from bounding the ratio of $s_1(x) = h_1(x) - \log(x)$, and $s_2(x) = h_2(x) - \log (x)$, based on the definition of $P_{F_1}^\ast$ and $P_{F_2}^\ast$.
\end{proof}

Next we show that the information-theoretic Algorithm
\ref{Alg:inf-sample-single-bidder} is optimal up to constants for MHR
distributions. We provide the proof of the
following theorem in Appendix \ref{sec:population-lb-examples}.

\begin{theorem}\label{thm:mhr-inf-sample-LB}
    Let $M$ be any DSIC and IR mechanism that takes as input a product 
  distribution $\vDtilde = \Dtilde_1 \times \cdots \times \Dtilde_n$. Then there
  exists a product distribution 
  $\vDstar = \Dstar_1 \times \cdots \times \Dstar_n$ such that $d_k(\Dstar_i, \Dtilde_i) \le \alpha$, $\Dstar_i$ is MHR for every $i$, and
  \[ \REV(M, \vDstar) \le (1 - \tilde{\Omega}(n \cdot \alpha)) \cdot \OPT(\vDstar). \]
\end{theorem}

\subsection{Analysis for regular distributions} \label{sec:population:regular}

  In this section we state the results for the performance of Algorithm 
\ref{Alg:inf-sample-single-bidder} for regular distributions and we provide a
proof sketch for the single-bidder case. The full proof of the following theorem
can be found in Appendix \ref{sec:multibidderUpperBound}.

\begin{theorem}\label{thm:regular-inf-sample}
     Let $\vDstar = \Dstar_1 \times \cdots \times \Dstar_n$ be a product 
  distribution where every $\Dstar_i$ is regular. Let also 
  $\vDtilde = \Dtilde_1 \times \cdots \Dtilde_n$ be any product distribution 
  such that for all $i \in [n]$ it holds that 
  $d_k(\Dstar_i, \Dtilde_i) \le \alpha_i$. If $\tilde{M}$ is the mechanism
  that Algorithm~\ref{Alg:inf-sample-single-bidder} outputs with input $\vDtilde$ then 
  it holds that
  \[ \REV(\tilde{M}, \vDstar) \ge \left(1 - 5 \cdot \sqrt{\sum_{i = 1}^n \alpha_i} \right) \cdot \OPT(\vDstar). \]
\end{theorem}

\begin{proof}[Proof sketch for $n = 1$.]

We first prove a general result that for two regular distributions $F$ and $\bar F$, such that $d_k(F, \bar F) \leq \alpha$, where $F(x)$ is stochastically dominated by $F(x)$ for $x \in \R_{+}$. The optimal revenue of these two distributions is close, formally  $\frac{\OPT({F})}{\OPT({\bar F})}  \geq 1-O(\sqrt{\alpha})$. The first key step replies on using the link function $h_r(x) = \frac{1}{1-F(x)}$ for regular distributions. Since $h_r(x)$ preserves the same monotonicity property as $F(x)$, we first derive a lower bound on $\bar h_r(x, \bar F)$ that is $\bar h_r(x, \bar F) \geq h_r(x, F) - \alpha h_r^2(x, F) $, using the fact that $d_k(F, \bar F) \leq \alpha$. This bound gives us useful constraints to discuss in different cases in the following part of the proof. Denote the corresponding optimal reserve prices for $F$ and $\bar F$ as $P$ and $\bar P$. We discuss separately two cases for $h(\bar P)$, where, for case 1 we have $h(\bar P) \leq \frac{1}{\sqrt{\alpha}}$, and for case 2, we have $h(\bar P) > \frac{1}{\sqrt{\alpha}}$. Using the connection from the link function to the revenue (see Figure~\ref{fig:link-function-plot}), case 1 directly leads to the conclusion that $ \frac{\OPT({F})}{\OPT({\bar F})} \geq 1 - \sqrt{\alpha}$. Case 2 is more subtle and requires a more careful argument.
Lastly, by construction, Algorithm~\ref{Alg:inf-sample-single-bidder} runs the Myerson optimal auction on a regular distribution $\hat F$, such that $\hat F \geq \hat F'(x)$ for all $x \in \R_{+}$, for any other regular distribution $F'(x)$ such that $d_k(F'(x), \tilde F(x)) \leq \alpha$. Applying the triangle inequality and combining with the conclusions  obtained from the two cases concludes the proof.
\end{proof}

Finally, we show that the information-theoretic Algorithm
\ref{Alg:inf-sample-single-bidder} is optimal up to constants for regular
distributions. We provide the proof of the
following theorem in Appendix \ref{sec:population-lb-examples}.

\begin{theorem}\label{thm:regular-inf-sample-LB}
    Let $M$ be any DSIC and IR mechanism that takes as input a product 
  distribution $\vDtilde = \Dtilde_1 \times \cdots \times \Dtilde_n$. Then there
  exists a product distribution 
  $\vDstar = \Dstar_1 \times \cdots \times \Dstar_n$ such that $d_k(\Dstar_i, \Dtilde_i) \le \alpha$, $\Dstar_i$ is regular for every $i$, and
  \[ \REV(M, \vDstar) \le (1 - \Omega(\sqrt{n \cdot \alpha})) \cdot \OPT(\vDstar). \]
\end{theorem}

%% file: sec/sample_complexity.tex
\section{Finite Samples} \label{sec:samples}

We provide a practical algorithm that takes samples from the 
corrupted distribution $\vDtilde$ as an input. We show that this algorithm achieves
almost optimal sample complexity for the MHR distribution case and the
single-bidder regular distribution case, whereas for the multi-bidder regular distributions there is a small gap between our upper and lower bounds.

An important notion to explain our algorithm for the finite-sample case is the 
following notion of the convex envelope.
\begin{definition}[Convex Envelope]
    The convex envelope $Conv(f)$ of a function $f$ is a function with the following
  property
  \[ Conv(f)(x)=\sup\{g(x)\mid g{\text{ is convex and }}g\leq f{\text{ over }}\R_{+}\}. \]
  In words, $Conv(f)$ is the maximum convex function that is below $f$.
\end{definition}

  For our algorithm one important property of the convex envelope is expressed in the
following lemma whose proof is presented in Appendix \ref{app:sec:tech-lemmas}.

\begin{lemma}
    Let $f$ be a non-decreasing piecewise constant function with $k$ pieces, then $Conv(f)$
  can be computed in time $\mathrm{poly}(k)$ and is a piecewise linear function
  with $O(k)$ pieces.\label{lem:compute-envelope}
\end{lemma}

\begin{algorithm}[H]
	\begin{algorithmic}[1]
	\caption{Robust Empirical Myerson Auction}
	\label{alg:emp-single-bidder}	
		\STATE \textbf{Input:} $m$ i.i.d.\ samples from (possibly corrupted) value distribution $\vD = \Pi_{i=1}^n \cD_i$, link function $h(\cdot)$.
		\STATE Let $\vE = \Pi_{i=1}^n E_i$ be the empirical distribution, i.e., the uniform distribution over the samples.\\\vspace{2mm}
		
		\FOR{$i=1\ldots n$} \vspace{2mm}
		
		\STATE Construct $\hat{E_i}$ as following: let $q^{E_i}(v)$ be the quantile of $E_i$; the quantile of $\hat{E_i}$ is as follows:
		\begin{equation*}
		q^{\hat{E}_i}(v)=
		\begin{cases}
		\max \left\{ 0, q^{E_i}(v) - \sqrt{\frac{2 q^{E_i}(v) \left( 1-q^{E_i}(v) \right) \ln (2mn \delta^{-1})}{m}} - \frac{4 \ln (2mn \delta^{-1})}{m} - \alpha_i\right\} & \text{if $v > 0$} \\
		1 & \text{if $v = 0$} 
		\end{cases}
		\end{equation*}
		
	    \STATE Construct $\tilde{E_i}$ such that $h\left(\tilde E_i(\cdot)\right)$ is the convex envelope of $h\left(\hat E(\cdot)\right)$ , i.e. 
	    \begin{equation*}
	    \tilde E_i(\cdot) = h^{-1}(Conv(\hat E_i(\cdot))
	    \end{equation*}
	    \ENDFOR \vspace{2mm}
	    \STATE Set $\tilde \vE = \Pi_{i=1}^n \tilde E_i$
		\STATE Output Myerson's optimal auction $M_{\tilde{\vE}}$ w.r.t.\ $\tilde{\vE}$.
	\end{algorithmic}
\end{algorithm}

The above algorithm resembles the main algorithm of \cite{guo2019settling} with
the addition of step 5. We first show that step 5 is necessary if we wish to 
obtain any non-trivial result in the robust auction learning setting that we
explore in this paper.

\begin{counterexample} \label{cex:necessity}
Imagine we have just one agent, i.e., $n = 1$,
with true distribution $\Dstar$ equal to an exponential distribution with 
parameter $\lambda = 1$. Also, to strengthen our counterexample imagine that we
have available an infinite number of samples, i.e., $m \to \infty$.
Now consider $\Dtilde$ to be the corrupted distribution where probability mass
$\alpha$ is removed from the mass closer to $0$ and it is placed as a point mass
at the point $c/\alpha$ for some number $c$. In this case, running Algorithm 
\ref{alg:emp-single-bidder} without step 5 will result is implementing an 
auction with reserve price that is very close to $c/\alpha$. The probability 
though that the true agent with distribution $\Dstar$ will buy this item goes to
zero with a rate $\exp(-c/\alpha)$ as $c \to \infty$. Hence, the total revenue
will be at most $(c/\alpha) \cdot \exp(-c/\alpha)$ and therefore we can make the 
total revenue to go to zero as we increase $c \to \infty$. Observe that this 
counterexample works even though we assumed that the initial distribution 
$\Dstar$ is MHR.
\end{counterexample}

We next provide the analysis of the performance of Algorithm 
\ref{alg:emp-single-bidder} for MHR and regular distributions. The proof of the following result can be found in Appendix~\ref{sec:sample-complexity-ub-proofs}.

\begin{theorem}[Finite samples, Regular distribution]\label{thm:regular-sample-UB}
    Let $\vDstar = \Dstar_1 \times \cdots \times \Dstar_n$ be a product 
  distribution where every $\Dstar_i$ is regular. Let also 
  $\vDtilde = \Dtilde_1 \times \cdots \Dtilde_n$ be any product distribution 
  such that for all $i \in [n]$ it holds that 
  $d_k(\Dstar_i, \Dtilde_i) \le \alpha_i$. If $\tilde{M}$ is the mechanism
  that Algorithm~\ref{alg:emp-single-bidder} outputs with input $m$ samples from
  $\vDtilde$ and assume that 
  $m = \tilde{\Omega}\left(\max_{i \in [n]} \left\{ \log(\frac{1}{\delta})/\alpha_i^{2} \right\} \right)$ then it holds that
  \[ \Pr\left( \REV(\tilde{M}, \vDstar) \ge \left(1 - O \left( \sqrt{\sum_{i = 1}^n \alpha_i} \right) \right) \cdot \OPT(\vDstar)\right) \ge 1 - \delta. \]
  Additionally, in the single-bidder case with $n = 1$ and $\alpha = \alpha_1$ the
  sample requirement becomes $m = \tilde{\Omega}\left(\log(\frac{1}{\delta})/\alpha^{3/2}\right)$.
\end{theorem}

The corresponding theorem for MHR distributions is the following, whose proof can be
found in Appendix \ref{sec:sample-complexity-ub-proofs}.

\begin{theorem}[Finite samples, MHR distribution]\label{thm:mhr-sample-UB}
    Let $\vDstar = \Dstar_1 \times \cdots \times \Dstar_n$ be a product 
  distribution where every $\Dstar_i$ is MHR. Let also 
  $\vDtilde = \Dtilde_1 \times \cdots \Dtilde_n$ be any product distribution 
  such that for all $i \in [n]$ it holds that 
  $d_k(\Dstar_i, \Dtilde_i) \le \alpha_i$. If $\tilde{M}$ is the mechanism
  that Algorithm~\ref{alg:emp-single-bidder} outputs with input $m$ samples from
  $\vDtilde$ and assume that 
  $m = \tilde{\Omega}\left(\max_{i \in [n]} \left\{ \log\left(\frac{1}{\delta}\right)/\alpha_i^{2} \right\} \right)$ then it holds that
  \[ \Pr\left( \REV(\tilde{M}, \vDstar) \ge \left(1 - \tilde{O} \left( \sum_{i = 1}^n \alpha_i \right) \right) \cdot \OPT(\vDstar)\right) \ge 1 - \delta. \]
\end{theorem}

We make a few remarks about the sample complexity upper bounds in the sequel. 

First, in both Theorem~\ref{thm:regular-sample-UB} and Theorem~\ref{thm:mhr-sample-UB}, the sample complexity upper bounds depend in a simple way on the sum of all the fractions of corruptions for each bidder; i.e., $\sum_{i = 1}^n \alpha_i$, indicating the important effect of the \emph{total} amount of corruption. Second, for regular distributions, in Theorem~\ref{thm:regular-sample-UB} we obtain a tight sample complexity bound for the single-bidder case, with $m = \tilde{\Omega}\left(\log(\frac{1}{\delta})/\alpha^{3/2}\right).$ For multi-bidder settings, our upper bound contains a small gap, with $m = \tilde{\Omega}\left(\max_{i \in [n]} \left\{ \log(\frac{1}{\delta})/\alpha_i^{2} \right\} \right)$. Whether such a gap can be matched is an interesting open question for future work. Lastly, comparing Theorem~\ref{thm:regular-sample-UB} and Theorem~\ref{thm:mhr-sample-UB}, it appears that for the multi-bidder settings the sample complexity bounds are of the same order, but we emphasize the key difference that for regular distributions this sample size is needed to provide a much \emph{weaker} guarantee on the revenue objective, which is a $ \left(1-O \left( \sqrt{\sum_{i = 1}^n \alpha_i} \right) \right)$ fraction of the optimal revenue, while the guarantee for MHR distributions is  a  $\left(1-O \left( \sum_{i = 1}^n \alpha_i \right) \right)$ fraction of the optimal revenue.

We next provide an information-theoretic lower bound that establishes the 
tightness of our upper bounds for the single-bidder single-item case with regular and MHR
distributions. 

\begin{theorem}[Sample complexity lower bounds]\label{thm:sample-LB}
Let $M$ be any DSIC and IR mechanism for a single-item 
single-buyer setting that takes as input $m$ samples from a
distribution $\Dtilde$. If 
\[ \REV(M, \Dstar) \ge (1 - O(\sqrt{\alpha})) \cdot \OPT(\Dstar), \]
for all distributions $\Dstar$ such that
$d_k(\Dstar, \Dtilde) \le \alpha$, where $\Dstar$ is 
regular, then
$m \ge \tilde \Omega\left(\log(\frac{2}{\delta})/\alpha^{3/2}\right)$. Additionally, if
\[ \REV(M, \Dstar) \ge (1 - O(\alpha)) \cdot \OPT(\Dstar), \]
for all distributions $\Dstar$ such that 
$d_k(\Dstar, \Dtilde) \le \alpha$, where $\Dstar$ is MHR, we have
$m \ge \tilde \Omega\left(\log(\frac{2}{\delta})/\alpha^{3/2}\right)$.
\end{theorem}

Theorem~\ref{thm:sample-LB} provides a general sample complexity lower bound on learning a near-optimal auction with at least a $ (1 - O(\sqrt{n \cdot \alpha}))$ fraction of the optimal revenue under the true valuation distribution. In comparison to our upper bounds (see Theorem~\ref{thm:regular-sample-UB} and Theorem~\ref{thm:mhr-sample-UB}), there is a small gap and we leave the nature of this gap as an open question for future work.

%% file: sec/conclu.tex
\section{Conclusions} \label{sec:conclu}

We have studied the learning of revenue-optimal auctions for multiple bidders, in a setting in which the samples can be corrupted adversarially. We first consider the information-theoretic limit in a population model, assuming  exact knowledge of the adversarially perturbed valuation distribution. We develop a theoretical algorithm which obtains a tight upper bound on the revenue for the MHR and regular distributions, obtaining the information-theoretic limit of the robustness guarantee. We then relax the population model and derive sample complexity bounds for learning optimal auctions from samples. We propose a practical algorithm which takes the corrupted samples as input, and provide the sample complexity upper bounds for the MHR distribution case and the single-bidder regular distribution case.  We also provide accompanying sample complexity lower bounds, and demonstrate a small gap relative to the corresponding upper bounds.

%% file: sec/proofs/tech_lemmas.tex
\section{Proofs of Technical Lemmas}\label{app:sec:tech-lemmas}

\begin{replemma}{lem:linkFunctions}
  A distribution with CDF $F$ is MHR if and only if $h_M(x; F)$ is a convex 
  function of $x$. Similarly, $F$ is regular if and only if $h_r(x; F)$ is a convex function of $x$. 
  Moreover, for two MHR (resp. regular) distributions $F_1$ and $F_2$, such that $F_1 \succeq F_2$, then we have that $h_M(x; F_1) \le h_M(x; F_2)$ (resp. $h_r(x; F_1) \le h_r(x; F_2)$) for all $x$.
\end{replemma}

\begin{proof}
We first show that given the CDF of any MHR distribution $F(x): \R_{+} \rightarrow [0,1]$, $h_M(x) \defeq -\log(1-F(x))$ is a convex, non-decreasing function with $h(0) = 0$.  (Without loss of generality, we consider $x \in [0, \infty]$, i.e. $\argmin_x h(x) = 0$.) We first present the analysis for the case when the distribution is continuous and smooth, and then generalize the same statement to discrete distributions. 

\noindent \underline{MHR continuous distributions:}\\
Denote the corresponding PDF of $F(x)$ as $f(x)$, and $g(x) \defeq \frac{f(v)}{1-F(v)}$. By definition, $F(0) = 0$ implies $h_M(0) = 0$.  Then, given that $F(x)$ is MHR, we have that $g(x)$ is monotone non-decreasing. By construction, 
\begin{align*}
   (h_M(x))''= \left(\frac{f(v)}{1-F(v)}\right)' = g'(x) \geq 0.
\end{align*}
Therefore, $h_M(x)$ is convex. Moreover, since $F(x)$ is a CDF thus non-decreasing, $h_M(x)=-\log(1-F(x))$ is also non-decreasing. We show that given any $h_M(x): \R_{+} \rightarrow \R_{+} $, such that $h_M(x)$ is convex, non-decreasing, $h_M(0) = 0$, and $\max_{x} h_M(x) = \infty$. Then, $F(x) \defeq 1-\exp(-h_M(x))$ is CDF of an MHR distribution.

By construction, $h_M(0) = 0$ implies $F(0)=0$, and $\max_x h_M(x)$ implies $\max_x F(x) =1$. Also given that $h_M(x)$ is convex, $g'(x) =  \left(\frac{f(v)}{1-F(v)}\right)' = (h_M(x))'' \geq 0$, which by definition implies $F(x)$ is MHR.

\noindent \underline{MHR discrete distributions:}\\
The lemma statement generalizes to the case when the valuation is discrete. We assume that the valuation can take a discrete set of values $\{x_i\}, i=1, \cdots, n$. Without loss of generality, we will restrict these values to the set $\mathbb{N}_0$ with probability mass function $P(x=i) = p_i; i =0\cdots n$. We define the \textit{discrete} hazard rate as: 
\[g(x_i) = \frac{P(x=i)}{P(x\geq i)}.\]
Then, the valuation distribution is MHR iff the discrete hazard rate is non-decreasing:
\begin{equation}\label{eq:mhr-discrete-def}
    g(x_{i+1}) \geq g(x_i),
\end{equation}
for all $i=0\cdots n$.

In this case, our link function will also be discrete. Further, denote $s_i \defeq P(x\geq i)$, then
\[
h(x_i) = -\log(P(x \geq x_i)) = -\log(s_i).
\]
Then $h(x)$ is convex if and only if for any $i\geq 0$,  
\begin{equation}\label{eq:mhr-discrete-h-convex}
h(x_{i+2}) - h(x_{i+1}) \geq h(x_{i+1} - h(x_i).
\end{equation}

We show that Eq~\eqref{eq:mhr-discrete-def} and Eq~\eqref{eq:mhr-discrete-h-convex} are equivalent. Notice that
\begin{align*}
   &h(x_{i+2}) - h(x_{i+1}) \geq h(x_{i+1} - h(x_i) \\
   &\iff \frac{s_{i+1}}{s_{i+1}-p_{i+1}} \geq \frac{s_i}{s_i-p_i} \\
   &\iff p_{i+1}s_i \geq p_i s_{i+1}\\
   &\iff \frac{p_{i+1}}{s_{i+1}} \geq \frac{p_i}{s_i}\\
    &\iff g(x_{i+1}) \geq g(x_i),
\end{align*}
which completes the proof.

\noindent \underline{Regular continuous distributions:}\\
We further prove a similar statement for regular distributions. First, given a CDF of a regular distribution $F(x)$, 
\[
\left(\frac{1}{1-F(x)}\right)'' = \frac{(1-F(x))f(x)'+ 2f(x)^2}{(1-F(x))^3}.
\]
By definition, the virtual value function is $\phi(x) \defeq v-\frac{1-F(x)}{f(x)}$, and
\[
\phi'(x) = \frac{(1-F(x))f(x)'+ 2f(x)^2}{f(x)^2}.
\]
Therefore, $\left(\frac{1}{1-F(x)}\right)''$ and $\phi'(x)$ share the same sign. Moreover, the distribution with CDF as $F(x)$ is regular if and only if the virtual value $\phi(x)$ is monotonically non-decreasing, which is $\phi'(x) \geq 0$. Hence the regularity of $F(x)$ implies that $h_r(x) \defeq \frac{1}{1-F(x)}$ is convex. Since $F(x)$ is a CDF thus non-decreasing, $h_r(x)=\frac{1}{1-F(x)}$ is also non-decreasing.

\noindent \underline{Regular discrete distributions:}\\
Similar to the MHR distributions, the lemma statement generalizes to the case when the valuation is discrete for regular distributions. Assume that the valuation can take a discrete set of values $\{x_i\}, i=1, \cdots, n$. Without loss of generality, we will restrict these values to the set $\mathbb{N}_0$ with probability mass function $P(x=i) = p_i; i =0\cdots n$. Further, consistent with the proof for MHR distributions, we denote $s_i \defeq P(x\geq i)$. 

The \textit{discrete} virtual value function is defined as:
\[
\phi(x_i) = x_i - \frac{s_{i}}{p_{i}},
\]
and the valuation distribution is regular iff $\phi(x)$ is non-decreasing:
\begin{equation}\label{eq:regular-discrete-def}
    \phi(x_{i+1}) \geq \phi(x_i),
\end{equation}
for all $i=0\cdots n$. 

In this case, our link function will again be discrete:
\[
h(x_i) = \frac{1}{P(x\geq x_i)} = \frac{1}{s_i}.
\]
and $h(x)$ is convex if and only if for any $i\geq 0$,  
\begin{equation}\label{eq:regular-discrete-h-convex}
h(x_{i+2}) - h(x_{i+1}) \geq h(x_{i+1}) - h(x_i).
\end{equation}
We show that Eq~\eqref{eq:regular-discrete-def} and Eq~\eqref{eq:regular-discrete-h-convex} are equivalent. 
\begin{align}\label{eq:regularDiscretePart1}
\begin{split}
    & h(x_{i+2}) - h(x_{i+1}) \geq h(x_{i+1}) - h(x_i) \\
    &\iff \frac{1}{s_{i+2}} + \frac{1}{s_i} \geq \frac{2}{s_{i+1}}
    \\
    &\iff \frac{1}{s_{i+1} - p_{i+1}} + \frac{1}{s_i} \geq \frac{2}{s_{i+1}}
    \\
    &\iff s_{i+1}^2 + p_ip_{i+1} \geq s_i s_{i+1} - s_i p_{i+1}. \\
    &\iff p_ip_{i+1} + p_{i+1}s_i + s_{i+1}(s_{i+1} - s_i) \geq 0\\
    &\iff p_ip_{i+1} + p_{i+1}s_i - s_{i+1}p_i \geq 0
    \end{split}
\end{align}
Moreover, from the regularity condition Eq~\eqref{eq:regular-discrete-def}, we have
\begin{align}\label{eq:regularDiscretePart2}
\begin{split}
&\phi(x_{i+1}) \geq \phi(x_i) \\
&\iff  i + 1 - \frac{s_{i+1}}{p_{i+1}} \geq i - \frac{s_i}{p_i} \\
&\iff 1 - \frac{s_{i+1}}{p_{i+1}} + \frac{s_i}{p_i} \geq 0 \\
&\iff p_ip_{i+1} + p_{i+1}s_i - s_{i+1}p_i \geq 0.
\end{split}
\end{align}
Combining \eqref{eq:regularDiscretePart1} and \eqref{eq:regularDiscretePart2} together completes the proof.

\noindent \underline{Stochastic dominance:}\\
Lastly, we show that for two MHR (resp. regular) distributions $F_1$ and $F_2$, such that $F_1 \succeq F_2$, then we have that $h_M(x; F_1) \le h_M(x; F_2)$ (resp. $h_r(x; F_1) \le h_r(x; F_2)$) for all $x$. This follows directly from the monotonicity of the link functions and the definition of stochastic dominance (see Definition~\ref{def:stochastic-dominance}).

Recall that the link function $h_M(x; F)$ for MHR distributions is defined as $h_M(x; F) = - \ln(1 - F(x))$, and the link function $h_r(x; F)$ for regular distributions is defined as $h_r(x; F) = 1/(1 - F(x))$. Therefore, for two MHR (resp. regular) distributions $F_1$ and $F_2$, $F_1(x) < F_2(x)$ implies $h_M(x, F_1) < h_M(x, F_2)$ (resp. $h_r(x, F_1) < h_r(x, F_2)$), which completes the proof.

\end{proof}

\begin{replemma}{lem:compute-envelope}
     Let $f$ be a non-decreasing piecewise constant function with $k$ pieces, then $Conv(f)$ can be computed in time $\mathrm{poly}(k)$ and is a piecewise linear function
  with $O(k)$ pieces.
\end{replemma}

\begin{proof}
Given that $f(x)$ is a non-decreasing piecewise constant function with $k$ pieces, we show that the following iterative procedure outputs its lower convex envelope $Conv(f)$
, which can be computed in time $\mathrm{poly}(k)$ and is a piecewise linear function with $O(k)$ pieces. Figure~\ref{fig:conv-env} provides an illustration of the construction according to this procedure.

\begin{protocol}[H]
\begin{algorithmic}[1]
\caption{Computing lower convex envelope for non-decreasing piecewise constant functions}
\STATE{\textbf{Input: a piecewise constant function $f(x): \R\rightarrow \R$ with $k$ pieces.} Denote the left starting point of each piece and the end point as $x_0, \ldots , x_k$.}
\STATE{\textbf{Initialize: }$i\gets 0, i' \gets 0$.}
\WHILE{$i\leq k-1$}
\STATE{$\bar x_{i'} \gets x_i, g(\bar x_{i'}) \gets f(x_i)$.}\label{alg:step:update-g}
\STATE{$i' \gets i'+1$.}
\STATE{Compute $i \gets \argmin_{i < j \leq k} \frac{f(x_j) - f(x_i)}{x_j - x_i} $.}\label{alg:step:conv-env}
\ENDWHILE
\STATE{$\bar x_{i'} \gets x_i, g(\bar x_{i'}) \gets f(x_i)$; $k' \gets i'$.}
\STATE{\textbf{Return: a piecewise linear function $g(x): \R\rightarrow \R$ with $k' < k$ pieces.} The left starting points of each piece and the end points are $\bar x_0, \ldots , \bar x_{i'}$, with the corresponding function values as specified in the procedure.}
\end{algorithmic}
\end{protocol}

\begin{figure*}[!ht]
\centering
\begin{tabular}{c} 
\includegraphics[width=0.5\textwidth]{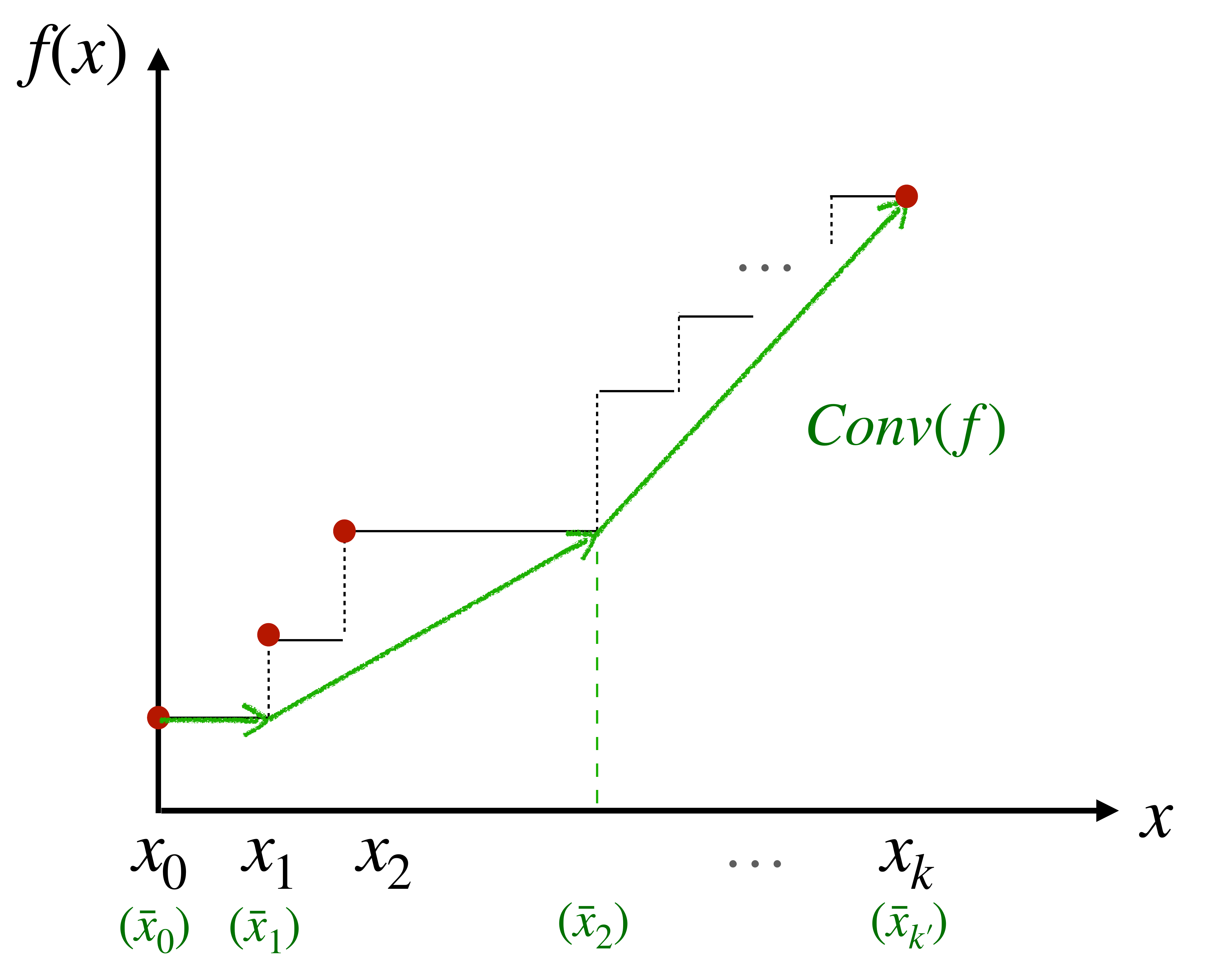} 
\end{tabular}
\caption{Lower convex envelope of a non-decreasing piecewise constant function $f(x)$ .}\label{fig:conv-env}
\end{figure*}

First, the above procedure requires at most $k^2$ rounds. We show that its output, $g(x)$, is the lower convex envelope for $f(x)$. It is clear from construction that $g(x)$ is piecewise linear, with vertices at $\bar x_0, \ldots , \bar x_{k'}$. Moreover, $g(x)\leq f(x)$ for all $x$ by construction.

Next we show that $g(x)$ is convex. Consider at a round $t$ with $i=i_t, 1 < 1 <k$. Then, step~\eqref{alg:step:conv-env} computes $i_{t+1} = \argmin_{i_t < j \leq k} \frac{f(x_j) - f(x_{i_t})}{x_j - x_{i_t}} $. Further denote $\min_{i_t < j \leq k} \frac{f(x_j) - f(x_{i_t})}{x_j - x_{i_t}}$ as $s(i_t)$. We show that $s(i_{t+1}) \geq s_{i_t}$. 

Suppose that $s(i_{t+1}) < s_{i_t}$. Then there exists $j^\ast > i_{t+1}> i_t$, such that 
\[
\frac{f(x_{j^\ast}) - f(x_{i_{t+1}})}{x_{j^\ast} - x_{i_{t+1}}} < \frac{f(x_{i_{t+1}}) - f(x_{i_t})}{x_{i_{t+1}} - x_{i_t}},
\]
which further implies that 
\[
\frac{f(x_{j^\ast}) - f(x_{i_{t}})}{x_{j^\ast} - x_{i_{t}}} < \frac{f(x_{i_{t+1}}) - f(x_{i_t})}{x_{i_{t+1}} - x_{i_t}}.
\]
Since $j^\ast > i_{t+1}> i_t$, this contradicts the fact that $i_{t+1} = \argmin_{i_t < j \leq k} \frac{f(x_j) - f(x_{i_t})}{x_j - x_{i_t}}$. Therefore  $s(i_{t+1}) \geq s_{i_t}$, which means that the slope of each piece for $g(x)$ is non-decreasing. Thus $g(x)$ is convex. Lastly, since $g(x)$ has all vertices with the same function values as $f(x)$, i.e. $g(x) = f(x)$ at all its vertices, and given that $g(x) \leq f(x)$ for all $x$, the values at these vertices are maximized and cannot be further improved. This completes the proof.
\end{proof}

We further provide two lemmas which present useful properties of the link functions in connection to the revenue.

\begin{lemma}\label{app:single-mhr-h-lem}
Given an MHR distribution with the CDF as $F(x): \R_{+} \rightarrow [0,1]$. Define $h(x) \defeq -\log(1-F(x))$. Then, at any reserve price $x$, the expected revenue $R(x) = \exp(-h(x) + \log(x))$. Moreover, the optimal reserve price $P^\ast_{F}$ is the minimizer of $(h(x) - \log(x))$.
\end{lemma}
\begin{proof}
First by construction, $h(x) - \log(x) = -\log(R(x))$. By definition, the optimal reserve price maximizes the revenue $R(x) = x(1-F(x))$, thus
\begin{align*}
    &\max \quad  x(1-F(x)) \\
    &\iff  \min  \quad  -\log(x(1-F(x)) ) \\
     &\iff  \min  \quad  - \log (x) - \log (1-F(x)) \\
    &\iff \min  \quad   h(x) - \log(x),
\end{align*}
which completes the proof.
\end{proof}

\begin{lemma}\label{app:lem:single-mhr-price-location}
Consider a valuation distribution $\cD$ with CDF as $F(x)$. Denote the optimal reserve price as $P^\ast_F$ and the optimal expected revenue at $P^\ast_F$ as $\OPT_F$. Then $P^\ast_F$ should be $P^\ast_F \leq e$, assuming that  $\OPT_F \leq 1$ and $F(x)$ is MHR.
\end{lemma}
\begin{proof}
By Lemma~\ref{app:single-mhr-h-lem}, $\OPT_F \leq 1$ implies that, 
\[
h(P_F^\ast) = \log(P_F^\ast) + b,
\]
for some $b \geq 0$. Also by Lemma~\ref{lem:linkFunctions}, $h$ is convex. Combined with the fact that $\OPT_F$ is the optimal reserve price and the concavity of $\log(x)$, $\OPT_F$ is the only point where $h(P_F^\ast) = \log(P_F^\ast) + b$ holds.

Now consider a linear function $y = ax, a>0$, which is a tangent line of the function $\log(x) + b$. Denote the tangent point as $x^\ast$. Solving the equation that $a = (\log(x))'= \frac{1}{x}$, and $ax = \log(x) + b$ give that:
\[
 x^\ast = e^{1-b} \leq e.
\]

Suppose that $P_{F}^\ast > x^\ast$. Consider the linear function $g(x) = \frac{h(P_{F}^\ast)}{P_{F}^\ast} x$.
Since $x^\ast$ is the tangent point, there exists a point $\bar x < P_{F}^\ast$, such that $g(\bar x) = \log(\bar x) + b$. Further, since $h$ is convex, for any point $0 < x < P_{F}^\ast$, we have $h(x) < g(x)$. By the continuity of $\log(x)$ and $h(x)$, there exists $\bar x' < P_{F}^\ast$, such that $h(\bar x') = \log(\bar x) + b$. This implies that $\bar x'$ achieves a larger revenue than $P_F^\ast$, and contradicts the fact that $P_F^\ast$ is the optimal reserve price. Hence, $P_F^\ast < x^\ast \leq e$, which completes the proof.
\end{proof}

%% file: sec/proofs/multibidder.tex
\section{Proof of Upper Bounds for the Population Model} \label{sec:multibidderUpperBound}

We first prove the following technical lemma that connects the coordinate 
Kolmogorov distance with the difference in expectation of increasing
functions.

\begin{definition}[Increasing Functions and Sets] \label{def:increasingFunctions}
    Let $u : \R^n \to \R$, we say that $u$ is increasing if for 
  every $\vv = (v_1, \dots, v_n)$, $\vv' = (v'_1, \dots, v'_n)$ such that
  $v'_i \ge v_i$, it holds that $u(\vv') \ge u(\vv)$. We say that the subset 
  $A \subseteq \R^n$ is increasing if and only if its characteristic function
  $\mathbf{1}_A(\vx)$ is an increasing function of $\vx$.
\end{definition}

\begin{lemma} \label{lem:increasingFunctionKolmogorov}
    Let $\vD = \D_1 \times \cdots \times \D_n$, 
  $\vD' = \D'_1 \times \cdots \times \D'_n$ be product $n$-dimensional 
  distributions with $d_k(\D_i, \D'_i) \le \alpha_i$. Then for every 
  increasing function $u : \R^n \to [0, \bar{u}]$ it holds that
  \[ \Abs{\E_{\vv \sim \vD}[u(\vv)] - \E_{\vv' \sim \vD'}[u(\vv')]} \le \bar{u} \cdot \left( \sum_{i = 1}^n \alpha_i \right). \]
\end{lemma}

\begin{proof}
    Our first step is to prove that the lemma holds for any function $u$ that
  is a characteristic function of an increasing set $A$ and then we extend to 
  all increasing functions.
  
    Let $u = \mathbf{1}_A$ we have that 
  $\E_{\vv \sim \vD}[u(\vv)] = \Pr_{\vv \sim \vD}(\vv \in A)$. We define the 
  sequence of distributions 
  $\vD_j = \D'_1 \times \cdots \times \D'_j \times \D_{j + 1} \times \cdots \times \D_n$ for $j = 0, \dots, n$,
  where obviously $\vD_0 = \vD$ and $\vD_n = \vD'$. Now via triangle inequality
  we have that
  \begin{align} \label{eq:subadditiveMultiBidderProof}
    \Abs{\Pr_{\vv \sim \vD}(\vv \in A) - \Pr_{\vv \sim \vD'}(\vv \in A)} \le \sum_{j = 1}^n \Abs{\Pr_{\vv \sim \vD_j}(\vv \in A) - \Pr_{\vv \sim \vD_{j - 1}}(\vv \in A)}.
  \end{align}
  \noindent Let $b_j(\vv_{-j})$ be the threshold of the step function 
  $\mathbf{1}_A(v_j, \vv_{-j})$ when we fix $\vv_{-j}$ and we view it as a 
  function of $v_j$. Now we have that
  \begin{align*}
    \Pr_{\vv \sim \vD_j}(\vv \in A) & = \int_{\R^n} \mathbf{1}_{A}(x_j, \vx_{-j}) ~~ d \D'_1(x_1) \cdots d \D'_j(x_j) \cdot d \D_{j + 1}(x_{j + 1}) \cdots d \D_n(x_n) \\
       & = \int_{\R^{n-1}} (1 - \D'_j(b_j(\vx_{-j}))) ~~ d \D'_1(x_1) \cdots d \D'_{j - 1}(x_{j - 1}) \cdot d \D_{j + 1}(x_{j + 1}) \cdots d \D_n(x_n)
  \end{align*}
  similarly we have
  \begin{align*}
    \Pr_{\vv \sim \vD_{j - 1}}(\vv \in A) & = \int_{\R^{n-1}} (1 - \D_j(b_j(\vx_{-j}))) ~~ d \D'_1(x_1) \cdots d \D'_{j - 1}(x_{j - 1}) \cdot d \D_{j + 1}(x_{j + 1}) \cdots d \D_n(x_n).
  \end{align*}
  Combining these we get that
  \begin{align*}
    & \Abs{\Pr_{\vv \sim \vD_{j}}(\vv \in A) - \Pr_{\vv \sim \vD_{j - 1}}(\vv \in A)} \le \\
    & ~~~~ \le \int_{\R^{n-1}} \Abs{\D'_j(b_j(\vx_{-j})) - \D_j(b_j(\vx_{-j}))} ~~ d \D'_1(x_1) \cdots d \D'_{j - 1}(x_{j - 1}) \cdot d \D_{j + 1}(x_{j + 1}) \cdots d \D_n(x_n).
  \end{align*}
  from the latter we can use the fact that $d_k(\D_j, \D'_j) \le \alpha_j$ and 
  we get that
  \begin{align*}
    \Abs{\Pr_{\vv \sim \vD_{j}}(\vv \in A) - \Pr_{\vv \sim \vD_{j - 1}}(\vv \in A)} \le \alpha_j.
  \end{align*}
  Applying the above to \eqref{eq:subadditiveMultiBidderProof} we get that
  \begin{align} \label{eq:KolmogorovIncreasingSetsDistance}
    \Abs{\Pr_{\vv \sim \vD}(\vv \in A) - \Pr_{\vv \sim \vD'}(\vv \in A)} \le \sum_{j = 1}^n \alpha_j.
  \end{align}
  
    The last steps is to extend the above to arbitrary increasing functions.
  We are going to approximate the increasing function $u$ via a sequence of
  functions $u_k$ which uniformly converges to $u$. Then we will show the 
  statement of the lemma for every function $u_k$ which by uniform convergence 
  implies the lemma for $u$ as well. We set 
  $A_{i, k} \triangleq \{\vx \in \R^n \mid u(\vx) \ge \frac{i}{k} \bar{u}\}$ and
  we define
  \[ u_k(\vx) = \frac{\bar{u}}{k} \sum_{i = 1}^k \mathbf{1}_{A_{i, k}}(\vx). \]
  Observe from the above definition that $u_k \to u$ uniformly and since $u$ is
  increasing we also have that all the sets $A_i$ are increasing. Also observe 
  that
  \[ \E_{\vv \sim \vD}[u_k(\vv)] = \frac{\bar{u}}{k} \sum_{i = 1}^k \Pr_{\vv \sim \vD}(\vv \in A_{i, k}) \]
  therefore we get that
  \begin{align*}
    \Abs{\E_{\vv \sim \vD}[u_k(\vv)] - \E_{\vv \sim \vD'}[u_k(\vv)]} \le 
    \frac{\bar{u}}{k} \sum_{i = 1}^k \Abs{\Pr_{\vv \sim \vD}(\vv \in A_{i, k}) - \Pr_{\vv \sim \vD'}(\vv \in A_{i, k})}.
  \end{align*}
  Now we can apply \eqref{eq:KolmogorovIncreasingSetsDistance} and we get
  \begin{align*}
    \Abs{\E_{\vv \sim \vD}[u_k(\vv)] - \E_{\vv \sim \vD'}[u_k(\vv)]} \le 
    \bar{u} \cdot \left( \sum_{j = 1}^n \alpha_j \right).
  \end{align*}
  Finally, since this is true for every $u_k$ and $u$ converges uniformly to $u$
  the above should be true for $u$ as well and hence the lemma follows.
\end{proof}

  We are going to use Lemma \ref{lem:increasingFunctionKolmogorov} both for the 
regular distributions case and for the MHR distributions case.

\subsection{Monotone Hazard Rate Distributions---Proof of Theorem \ref{thm:mhr-inf-sample}} \label{sec:multibidderUpperBound:MHR}

  In this section we show the part of the Theorem 
\ref{thm:mhr-inf-sample} related to $n > 1$. For the stronger result
for the case $n = 1$ we refer to Section \ref{sec:single-mhr-population-ub}.

  Let $\tilde{\vD}$ be the corrupted product distribution that we observe,
$\hat{\vD}$ be the output distribution of 
Algorithm~\ref{Alg:inf-sample-single-bidder}, $\vDstar$ be the original
distribution that we are interested in. We know from the description of
Algorithm~\ref{Alg:inf-sample-single-bidder} for 
$\hat{\vD} = \hat{\D}_1 \times \cdots \times \hat{\D}_n$ that $\hat{\D}_i$ is 
MHR, that $d_k(\hat{\D}_i, \Dstar_i) \le \alpha_i$ and that
$\hat{\D}_i \preceq \Dstar_i$. We also know that $\Dstar_i$ is MHR. Finally, we 
know that the output $M$ of Algorithm \ref{Alg:inf-sample-single-bidder} is the 
Myerson optimal mechanism for the distribution $\hat{\vD}$ and hence
$\REV(M, \hat{\vD}) = \OPT(\hat{\vD})$. So applying the strong revenue 
monotonicity lemma \ref{lem:strong-rev-monotone} we have that
\begin{align} \label{eq:revenueToOptimal}
  \OPT(\hat{\vD}) = \REV(M, \hat{\vD}) \le \REV(M, \vDstar).
\end{align}
Therefore to show Theorem \ref{thm:mhr-inf-sample}, it suffices to show that
\begin{align} \label{eq:optimalToOptimalCondition}
  \OPT(\hat{\vD}) \ge \left(1 - \tilde{O}\left(\sum_{i = 1}^n \alpha_i\right)\right) \cdot \OPT(\vDstar).
\end{align}

  We are going to use the following result from \cite{cai2011extreme} but with 
the formulation obtained in Lemma 17 of \cite{guo2019settling}, combined with 
the weak revenue monotonicity (Lemma 3 of \cite{guo2019settling}).

\begin{theorem}[\cite{cai2011extreme}] \label{thm:CaiDaskalakis}
    For any product MHR distribution $\vD$, and any 
  $\frac{1}{4} \ge \varepsilon \ge 0$ and 
  $u \ge c \cdot \log \left(\frac{1}{\varepsilon}\right) \OPT(\vD)$. Let
  $t_{u}(\D_1)$, $\dots$, $t_{u}(\D_n)$ be the distributions obtained by
  truncating $\D_1$, $\dots$, $\D_n$ at the value $\bar{u}$ and let $t_{u}(\vD)$
  be their product distribution, where $c$ is an absolute constant. Then, we
  have that
  \[ \OPT(\vD) \ge \OPT(t_{u}(\vD)) \ge \left( 1 - \varepsilon \right) \cdot \OPT(\vD). \]
\end{theorem}

  Now let 
$\bar{u} = c \cdot \log \left(\frac{1}{\varepsilon}\right) \OPT(\vDstar)$, then
we also have that
$\bar{u} \ge c \cdot \log \left(\frac{1}{\varepsilon}\right) \OPT(\hat{\vD})$ 
due to weak revenue monotonicity (Lemma 3 of \cite{guo2019settling}). Hence, 
applying Theorem \ref{thm:CaiDaskalakis} we have that
\begin{align} \label{eq:truncationEffect}
  \OPT(\hat{\vD}) \ge \OPT(t_{\bar{u}}(\hat{\vD})) \quad \quad \text{and} \quad \quad \OPT(t_{\bar{u}}(\vDstar)) \ge (1 - \varepsilon) \cdot \OPT(\vDstar).
\end{align}
Since we know that $d_k(\hat{\D}_i, \Dstar_i) \le \alpha_i$ we also have that
$d_k(t_{\bar{u}}(\hat{\D}_i), t_{\bar{u}}(\Dstar_i)) \le \alpha_i$. Let now 
$M^*_{\bar{u}}$ be the optimal mechanism for the distribution 
$t_{\bar{u}}(\vDstar)$. It is easy to see that the ex-post revenue obtained from
the mechanism $M^*_{\bar{u}}$ is an increasing function of the observed bids. 
Hence, we can apply Lemma \ref{lem:increasingFunctionKolmogorov} to the 
$[0, \bar{u}]$ bounded distributions $t_{\bar{u}}(\hat{\vD})$ and 
$t_{\bar{u}}(\vDstar)$ and we get that
\begin{align} 
  \OPT(t_{\bar{u}}(\hat{\vD})) \ge \REV(M^*_{\bar{u}}, t_{\bar{u}}(\hat{\vD})) & \ge \REV(M^*_{\bar{u}}, t_{\bar{u}}(\vDstar)) - \bar{u} \cdot \left(\sum_{i = 1}^n \alpha_i\right) \nonumber \\
  & = \OPT(t_{\bar{u}}(\vDstar)) - \bar{u} \cdot \left(\sum_{i = 1}^n \alpha_i\right). \label{eq:corruptionEffect}
\end{align}

  If we combine \eqref{eq:truncationEffect} and \eqref{eq:corruptionEffect} then
we have that
\begin{align} \label{eq:allEffects}
  \OPT(\hat{\vD}) \ge (1 - \varepsilon) \cdot \OPT(\vDstar) - \bar{u} \cdot \left(\sum_{i = 1}^n \alpha_i\right).
\end{align}
Now we can substitute the value of $\bar{u}$ to the above inequality and we get
that
\[ \OPT(\tilde{\vD}) \ge \left(1 - c \cdot \log\left(\frac{1}{\eps}\right) \cdot \left( \sum_{i = 1}^n \alpha_i \right) - \varepsilon  \right) \cdot \OPT(\vD). \]
Finally, setting $\varepsilon = \sum_{i = 1}^n \alpha_i$ we get
\[ \OPT(\tilde{\vD}) \ge \left(1 - (c + 1) \cdot \left( \sum_{i = 1}^n \alpha_i \right) \cdot \log \left( \frac{1}{\sum_{i = 1}^n \alpha_i} \right) \right) \cdot \OPT(\vD). \]
Hence, \eqref{eq:optimalToOptimalCondition} follows and as we explained this
proves Theorem \ref{thm:mhr-inf-sample}.

\subsection{Regular Distributions---Proof of Theorem \ref{thm:regular-inf-sample}} \label{sec:multibidderUpperBound:Regular}

  Let $\tilde{\vD}$ be the corrupted product distribution that we observe,
$\hat{\vD}$ be the output distribution of 
Algorithm~\ref{Alg:inf-sample-single-bidder}, $\vDstar$ be the original
distribution that we are interested in. We know from the description of
Algorithm~\ref{Alg:inf-sample-single-bidder} for 
$\hat{\vD} = \hat{\D}_1 \times \cdots \times \hat{\D}_n$ that $\hat{\D}_i$ is 
a regular distribution, that $d_k(\hat{\D}_i, \Dstar_i) \le \alpha_i$ and that
$\hat{\D}_i \preceq \Dstar_i$. We also know that $\Dstar_i$ is regular. Finally,
we know that the output $M$ of Algorithm \ref{Alg:inf-sample-single-bidder} is
the Myerson optimal mechanism for the distribution $\hat{\vD}$ and hence
$\REV(M, \hat{\vD}) = \OPT(\hat{\vD})$. So applying the strong revenue 
monotonicity lemma \ref{lem:strong-rev-monotone} we have that
\begin{align} \label{eq:revenueToOptimal:regular}
  \OPT(\hat{\vD}) = \REV(M, \hat{\vD}) \le \REV(M, \vDstar).
\end{align}
Therefore to show Theorem \ref{thm:regular-inf-sample}, it suffices to show that
\begin{align} \label{eq:optimalToOptimalCondition:regular}
  \OPT(\hat{\vD}) \ge \left(1 - \tilde{O}\left(\sum_{i = 1}^n \alpha_i\right)\right) \cdot \OPT(\vDstar).
\end{align}

  We are going to use the following theorem from \cite{devanur2016sample},
combined with the weak revenue monotonicity (Lemma 3 of \cite{guo2019settling}).
  
\begin{theorem}[Lemma 2 of \cite{devanur2016sample}] \label{thm:DevanurHuangPsomas}
    Let $\vD$ be a product of $n$ regular distributions and $\OPT(\vD)$ be the 
  optimal revenue of $\vD$. Suppose $\frac{1}{4} \ge \varepsilon \ge 0$ and 
  $u \ge \frac{1}{\varepsilon} \OPT(\vD)$. Let $t_{u}(\D_1)$, $\dots$, 
  $t_u(\D_n)$ be the distributions obtained by truncating $\D_1$, $\dots$, 
  $\D_n$ at the value $u$ and let $t_u(\vD)$ be their product distribution.
  Then, we have that
  \[ \OPT(\vD) \ge \OPT(t_u(\vD)) \ge (1 - 4 \varepsilon) \cdot \OPT(\vD). \]
\end{theorem}

  Now let 
$\bar{u} = \frac{1}{\varepsilon} \OPT(\vDstar)$, then we also have that
$\bar{u} \ge \frac{1}{\varepsilon} \OPT(\hat{\vD})$ due to weak revenue
monotonicity (Lemma 3 of \cite{guo2019settling}). Hence, applying Theorem
\ref{thm:DevanurHuangPsomas} we have that
\begin{align} \label{eq:truncationEffect:regular}
  \OPT(\hat{\vD}) \ge \OPT(t_{\bar{u}}(\hat{\vD})) \quad \quad \text{and} \quad \quad \OPT(t_{\bar{u}}(\vDstar)) \ge (1 - \varepsilon) \cdot \OPT(\vDstar).
\end{align}
Since we know that $d_k(\hat{\D}_i, \Dstar_i) \le \alpha_i$ we also have that
$d_k(t_{\bar{u}}(\hat{\D}_i), t_{\bar{u}}(\Dstar_i)) \le \alpha_i$. Let now 
$M^*_{\bar{u}}$ be the optimal mechanism for the distribution 
$t_{\bar{u}}(\vDstar)$. It is easy to see that the ex-post revenue obtained from
the mechanism $M^*_{\bar{u}}$ is an increasing function of the observed bids. 
Hence, we can apply Lemma \ref{lem:increasingFunctionKolmogorov} to the 
$[0, \bar{u}]$ bounded distributions $t_{\bar{u}}(\hat{\vD})$ and 
$t_{\bar{u}}(\vDstar)$ and we get that
\begin{align} 
  \OPT(t_{\bar{u}}(\hat{\vD})) \ge \REV(M^*_{\bar{u}}, t_{\bar{u}}(\hat{\vD})) & \ge \REV(M^*_{\bar{u}}, t_{\bar{u}}(\vDstar)) - \bar{u} \cdot \left(\sum_{i = 1}^n \alpha_i\right) \nonumber \\
  & = \OPT(t_{\bar{u}}(\vDstar)) - \bar{u} \cdot \left(\sum_{i = 1}^n \alpha_i\right). \label{eq:corruptionEffect:regular}
\end{align}

  If we combine \eqref{eq:truncationEffect:regular} and
\eqref{eq:corruptionEffect:regular} then we have that
\begin{align} \label{eq:allEffects:regular}
  \OPT(\hat{\vD}) \ge (1 - \varepsilon) \cdot \OPT(\vDstar) - \bar{u} \cdot \left(\sum_{i = 1}^n \alpha_i\right).
\end{align}
Now we can substitute the value of $\bar{u}$ to the above inequality and we get
that
\[ \OPT(\tilde{\vD}) \ge \left(1 - \frac{1}{\varepsilon} \cdot \left( \sum_{i = 1}^n \alpha_i \right) - 4 \varepsilon  \right) \cdot \OPT(\vD). \]
Finally, setting $\varepsilon = \sqrt{\sum_{i = 1}^n \alpha_i}$ we get
\[ \OPT(\tilde{\vD}) \ge \left(1 - 5 \cdot \sqrt{ \sum_{i = 1}^n \alpha_i }  \right) \cdot \OPT(\vD). \]
Hence, \eqref{eq:optimalToOptimalCondition:regular} follows and as we explained
this proves Theorem \ref{thm:regular-inf-sample}.

%% file: sec/proofs/PM_UB_MHR_single.tex
\subsection{MHR Distributions -- Proof of Theorem~\ref{thm:mhr-inf-sample}, $n = 1$ Case} \label{sec:single-mhr-population-ub}

In this subsection we show the part of the Theorem 
\ref{thm:mhr-inf-sample} related to $n = 1$, for which we obtain a stronger result compared to the case $n > 1$. We first show a useful proposition:

\begin{proposition}\label{prop:single-mhr-ball}
Consider two MHR distributions $\cD_1$, $\cD_2$ with CDFs as $F_1$ and $F_2$, such that $d_k(\cD_1, \cD_1) \leq \alpha$, and $F_1(x) \geq F_2(x)$ for all $x \in \R_{+}$. Denote the optimal expected revenue under $\cD_1$ and $\cD_2$ as $\OPT_{F_1}$ and $\OPT_{F_2}$, and the corresponding optimal reserve prices as $P_{F_1}^\ast$ and $P_{F_2}^\ast$. Then,
\[
\left(1+\alpha e\right)^{-1} \leq \frac{\OPT_{F_1}}{\OPT_{F_2}}  \leq 1+\alpha e.
\]
\end{proposition}
\begin{proof}
Consider two MHR distributions $\cD_1$, $\cD_2$ with CDFs as $F_1$ and $F_2$, such that $d_k(\cD_1, \cD_1) \leq \alpha$, and $F_1(x) \geq F_2(x)$ for all $x \in \R_{+}$. Denote the optimal expected revenue under $\cD_1$ and $\cD_2$ as $\OPT_{F_1}$ and $\OPT_{F_2}$, and the corresponding optimal reserve prices as $P_{F_1}^\ast$ and $P_{F_2}^\ast$. Without loss of generality, we consider $\OPT_{F_1} \geq \OPT_{F_2}$. Further, since the ratio of the revenues, e.g. $\frac{\OPT_{F_1}}{\OPT_{F_2}}$ is scale invariant, we assume without loss of generality that $\OPT_{F_1} = 1$.

By Lemma~\ref{app:lem:single-mhr-price-location}, we have $P_{F_1}^\ast \leq e$. By Lemma~\ref{app:single-mhr-h-lem}, $\OPT_{F_1}= 1$ implies that $h_1(P_{F_1}^\ast) = \log(P_{F_1}^\ast)$. Since $P_{F_1}^\ast \leq e$, we have
\begin{align*}
    &h_1(P_{F_1}^\ast) \leq 1 \\
    &\iff -\log(1-F_1(P_{F_1}^\ast)) \leq 1 \\
    &\iff F_1(P_{F_1}^\ast) \leq 1-\frac{1}{e} \\
    &\iff 1-F_1(P_{F_1}^\ast)  \geq \frac{1}{e}.
\end{align*}

Therefore, since $F_1$ is non-decreasing, for any $x < P_{F_1}^\ast$, $1-F_1(x) \geq \frac{1}{e}$. So for any $x < P_{F_1}^\ast$, we have
\begin{align*}
    |h_1(x) - h_2(x)|&= \left|\log \left(\frac{1-F_2(x)}{1-F_1(x)}\right)\right|\\
    &= \left|\log \left(1+\frac{F_1(x)-F_2(x)}{1-F_1(x)}\right)\right|\\
    &\leq \log \left(1+\alpha e\right)\\
    &= O(\alpha),
\end{align*}
where the at the second last step, the inequality follows from the fact that $d_k(\cD_1, \cD_1) \leq \alpha$, and $x < P_{F_1}^\ast$.

Further, $F_1(x) \geq F_2(x)$ for all $x \in \R_{+}$ implies that $h_1(x) \geq h_2(x)$ for all $x \in \R_{+}$. Therefore,  $h_1(P_{F_1}^\ast) = \log(P_{F_1}^\ast) \geq h_2(P_{F_1}^\ast)$. Therefore, we have $P_{F_2}^\ast \leq P_{F_1}^\ast$, and
\[
|h_1(P_{F_2}^\ast) - h_2(P_{F_2}^\ast)| \leq  \log \left(1+\alpha e\right).
\]
Now define functions $s_1(x) = h_1(x) - \log(x)$, and $s_2(x) = h_2(x) - \log (x)$. Then by the definition of $P_{F_1}^\ast$, $P_{F_2}^\ast$ and Lemma~\ref{app:single-mhr-h-lem}, 
\begin{align*}
    \min_{x \leq P_{F_1}^\ast} s_1(x) &= s_1(P_{F_1}^\ast) \leq s_1(P_{F_2}^\ast)\\
    &\leq s_2(P_{F_2}^\ast) +  \log\left(1+\alpha e\right)\\
    &=  \min_{x \leq P_{F_2}^\ast} s_2(x)  +  \log\left(1+\alpha e\right).
\end{align*}
Therefore, by the definitions of $s_1$ and $s_2$,
\begingroup
\begin{align*}
    &\left|\min_{x \leq P_{F_1}^\ast} s_1(x) -  \min_{x \leq P_{F_2}^\ast} s_2(x)\right| \leq \log\left(1+\alpha e\right) \\
    &\iff |\log(\OPT_{F_2}) - \log(\OPT_{F_1})| \leq \log\left(1+\alpha e\right)\\
     &\iff -\log\left(1+\alpha e\right) \leq \log(\OPT_{F_2}) \leq \log\left(1+\alpha e\right)\\
     &\iff \left(1+\alpha e\right)^{-1} \leq \OPT_{F_2} \leq 1+\alpha e.
\end{align*}
\endgroup
The above directly implies:
\[
 \left(1+\alpha e\right)^{-1} \leq \frac{\OPT_{F_1}}{\OPT_{F_2}}  \leq 1+\alpha e.
\]
which completes the proof.
\end{proof}

Now we are ready to prove Theorem~\ref{thm:mhr-inf-sample} for the $n=1$ case.

\begin{proof}
First, by construction, Algorithm~\ref{Alg:inf-sample-single-bidder} runs the Myerson optimal auction on an MHR distribution $\hat F$, such that $\hat F \geq \hat F'(x)$ for all $x \in \R_{+}$, for any MHR distribution $F'(x)$ such that $d_k(F'(x), \tilde F(x)) \leq \alpha$. Also by assumption, $d_k(F^\ast(x), \tilde F(x)) \leq \alpha$. Therefore by triangle inequality, $d_k(F^\ast(x), \hat F(x)) \leq d_k(F^\ast(x), \tilde F(x)) + d_k(\tilde F(x), \hat F(x)) \leq 2\alpha$.

Denote $\alpha' = 2\alpha$. By Proposition~\ref{prop:single-mhr-ball},
\[
\left(1+\alpha' e\right)^{-1} \leq \frac{\OPT_{F_1}}{\OPT_{F_2}}  \leq 1+\alpha' e.
\]
Note that $ \left(1+\alpha' e\right)^{-1} =  \left(1+2 \alpha e\right)^{-1} = 1-O(\alpha)$, which completes the proof.
\end{proof}

%% file: sec/proofs/PM_LB_examples.tex
\section{Proof of Optimality for the Upper Bounds}\label{sec:population-lb-examples}

  For these lower bounds we follow the idea of the lower bounds from 
\cite{guo2019settling} adapted to the corrupted case that we consider in this 
paper. The lower bound constructions of \cite{guo2019settling} are based on
a family of distributions
\[ \cH = \{ \vD \mid \D_1 = \D^b, \D_i = \D^h \quad \text{or} \quad \D_i = \D^\ell \text{ for all } 2 \le i \le n \}. \]
Observe that this family is characterized by the triplet of distributions 
$\D^b$, $\D^h$, and $\D^{\ell}$ for which we ask for the following conditions.
\begin{enumerate}
  \item[a)] $\D^b$ is a point mass at $v_0$.
  \item[b)] The propability of $v \ge v_2$ is at most $1/n$ both when 
            $v \sim \D^h$ and when $v \sim \D^{\ell}$.
  \item[c)] The probability of $v_1 > v \ge v_2$ is at least $p$ both when 
            $v \sim \D^h$ and when $v \sim \D^{\ell}$.
  \item[d)] For any value $v$ such that $v_1 > v \ge v_2$, we have
            $\phi^{\ell}(v) + \Delta \le v_0 \le \phi^{h}(v) - \Delta$, where
            $\phi^{\ell}$ is the virtual value function of $\D^{\ell}$ and 
            correspondingly for $\phi^h$.
  \item[e)] For any value $v$ such that $v < v_2$, we have that 
            $\phi^{h}(v), \phi^{\ell}(v) \le v_0$.
  \item[f)] For any value $v_1 > v \ge v_2$ we have that the ratio
            $\frac{d \D^h}{d \D^{\ell}}(v)$ is upper and lower bounded by a 
            constant, where $\frac{d \D^h}{d \D^{\ell}}$ is the Radon–Nikodym
            derivative between $\D^h$ and $\D^{\ell}$.
  \item[g)] $\D^{h}$ is regular.
  \item[h)] The point $v_1$ is either $+ \infty$ or is a point mass and an
            upper bound on the support in both $\D^{\ell}$ and $\D^h$.
\end{enumerate}

  Under these conditions and using the exact same proof as the Lemma 18 from
\cite{guo2019settling} we can show the following.

\begin{lemma} \label{lem:inf:lowerBound}
    Let $\cH$ be a class of distributions that satisfies the conditions a) - 
  h) and additionally satisfies the following.
  \begin{enumerate}
    \item[i)] We have that $d_k(\D^{\ell}, \D^h) \le \alpha/n$.
  \end{enumerate}
  Then any algorithm that is robust to a total corruption $\alpha$ in Kolmogorov
  distance across all bidders achieves revenue of at most
  \[ \OPT(\vD) - \Omega(n \cdot p \cdot \Delta) \]
  for any distribution $\vD \in \cH$.
\end{lemma}

\subsection{MHR Distributions -- Proof of Theorem~\ref{thm:mhr-inf-sample-LB}}\label{app:mhr-inf-sample-LB}

  Let $a = \ln(n) - \ln(1 - \beta)$, $b = \ln(n)$, $v_0 = a - 1$, 
$v_1 = \ln(n) - 2 \cdot \ln(1 - \beta)$, $v_2 = a$, 
$p = \beta \cdot (1 - \beta) / n$, $\Delta = 1/2$. Then we define $\D^{\ell}$
and $\D^h$ according to their CDFs $F^{\ell}$ and $F^h$ which are the 
following:
\begin{equation*}
  F^{\ell}(v) = 
  \begin{cases}
  1 - \exp(-v) & v < v_1 \\
  0            & v \ge v_1
  \end{cases},
\end{equation*}
\begin{equation*}
  F^{h}(v) = 
  \begin{cases}
  1 - \exp\left(-\frac{b}{a} \cdot v\right) & v < v_2 \\
  1 - \exp\left(-\frac{v_1 - b}{v_1 - a} \cdot (v - a) + b\right) & v_2 \le v < v_1 \\
  0            & v \ge v_1
  \end{cases}.
\end{equation*}
Observe also that for this choice of distributions it holds that
\begin{equation*}
  \phi^{\ell}(v) = 
  \begin{cases}
  v - 1 & v < v_1 \\
  v_1            & v \ge v_1
  \end{cases},
\end{equation*}
\begin{equation*}
  \phi^{h}(v) = 
  \begin{cases}
  v - \frac{a}{b} & v < v_2 \\
  v - \frac{v_1 - a}{v_1 - b} & v_2 \le v < v_1 \\
  v_1            & v \ge v_1
  \end{cases}.
\end{equation*}
Now the conditions a) - h) are easy to verify. For the condition i) we observe
that the maximum difference between the two CDFs is at $v = v_2$ for which
we have that $\Abs{F^{\ell}(v_2) - F^h(v_2)} \le \beta/n$.

Hence, Lemma \ref{lem:inf:lowerBound} implies that the maximum revenue 
achievable by any robust mechanism is
\[ \OPT(\vD) - \Omega(n \cdot p \cdot \Delta) = \OPT(\vD) - \Omega(\beta). \]
Observe that since the maximum value of any bidder is at most $\ln(n)$
we have that the maximum revenue is
\[ \left(1 - \frac{\beta}{\ln(n)}\right) \cdot \OPT(\vD). \]
If we write this expression with respect to the amount of corruption per bidder,
then we have that the maximum possible revenue is
\[ \left(1 - \frac{n \cdot \alpha}{\ln(n)}\right) \cdot \OPT(\vD). \]
Finally, we observe that all of $\D^b$, $\D^{\ell}$, and $\D^h$ are MHR and 
hence Theorem \ref{thm:mhr-inf-sample-LB} follows.

\subsection{Regular Distributions -- Proof of Theorem~\ref{thm:regular-inf-sample-LB}}\label{app:regular-inf-sample-LB}

  For the case of regular distributions we will use the same distributions 
used by \cite{guo2019settling} in their proof of their Theorem 2. In 
particular, let $v_0 = 3/2$, $v_1 = + \infty$, $v_2 = 1 + \frac{1}{\beta}$,
$p = \frac{\beta}{n}$, and $\Delta = 1/2$. We define $\D^{\ell}$ and $\D^h$ 
through their CDFs as follows
\begin{equation*}
  F^{\ell}(v) = 1 - \frac{1}{n \cdot (v - 1)},
\end{equation*}
\begin{equation*}
  F^{h}(v) = 
  \begin{cases}
  0 & v < 1 + \frac{1}{n} \\
  1 - \frac{1}{n \cdot (v - 1)} & 1 + \frac{1}{n} \ge v < v_2 \\
  1 - \frac{1 - \beta}{n \cdot (v - 2)} & v \ge v_2
  \end{cases}.
\end{equation*}
The fact that these distributions satisfy a) - h) can be found in 
\cite{guo2019settling}. We will focus on proving i). It is not hard to see
that the two CDFs appears when $v = \bar{v} = 1 + \frac{1}{\sqrt{1 - \beta}}$.
For this value we have
\[ \Abs{F^{\ell}(\bar{v}) - F^h(\bar{v})} = \frac{1}{n}\left(2 - \beta - 2 \sqrt{1 - \beta}\right) \le \frac{\beta^2}{n}, \]
where the last inequality can be easily verifies for $\beta \le 1$. Now 
setting $\alpha = \frac{\beta^2}{n}$, observing that 
$n \cdot p \cdot \Delta = \Omega(\beta)$, and observing that 
$\OPT(\vD) \le O(1)$ we can apply Lemma \ref{lem:inf:lowerBound} and we get 
that the maximum possible revenue is
\[  \left(1 - \Omega\left(\sqrt{n \cdot \alpha}\right) \right) \cdot \OPT(\vD). \]
Finally by observing that all of $\D^b$, $\D^{\ell}$, and $\D^h$ are regular 
Theorem \ref{thm:regular-inf-sample-LB} follows. 

%% file: sec/proofs/sample_UB_proofs.tex
\section{Proofs of Sample Complexity Bounds}\label{sec:sample-complexity-ub-proofs}

\subsection{Proof of Theorem~\ref{thm:regular-sample-UB}, $n>1$ Case}\label{app:multi-bidder-regular-sample-ub}

  This follows easily from Theorem \ref{thm:regular-inf-sample} and the DKW 
inequality \citet{dvoretzky1956asymptotic, massart1990tight} that states that the
empirical CDF with $m$ samples is close to the population CDF with an error of at
most
\[ O\left( \sqrt{\frac{\log(1/\delta)}{m}} \right) \]
with probability at least $1 - \delta$.
\qed

\subsection{Proof of Theorem~\ref{thm:regular-sample-UB}, $n=1$ Case}\label{app:single-bidder-regular-sample-ub}

We present in this section a proof of Theorem~\ref{thm:regular-sample-UB} for the case with $n=1$ and regular distributions. In this case, we show that Algorithm~\ref{alg:emp-single-bidder} achieves the optimal sample complexity, up to a poly-logarithmic factor.

First, by [Lemma 5,~\citet{guo2019settling}], we have that with probability at least $1-\delta$, for any value $v \geq 0$, the quantiles of $\Dtilde$ and its empirical counterpart $E$ satisfy that:
\begin{equation}\label{lem:regular-single-bernstein}
    |q^{E}(v) - q^{\Dtilde}(v)| \leq \sqrt{\frac{2q^{\Dtilde}(v)(1-q^{\Dtilde(v)})\ln(2m\delta^{-1})}{m}} + \frac{\ln(2m\delta^{-1})}{m}.
\end{equation}

Further note that by construction, we have
\[
q^{E} - q^{\hat E} \leq  \sqrt{\frac{2 q^{E}(v) \left( 1-q^{E}(v) \right) \ln (2m \delta^{-1})}{m}} + \frac{4 \ln (2m \delta^{-1})}{m} + \alpha.
\]

Given that Algorithm~\ref{alg:emp-single-bidder} runs the Myerson optimal auction on $\tilde E$, which is a minimal regular distribution that dominates $\tilde E$. Further,  $\hat E \succeq D^\ast$ by construction, assuming Eq~\eqref{lem:regular-single-bernstein} holds. Therefore, we have  $D^\ast \succeq \tilde E$ assuming Eq~\eqref{lem:regular-single-bernstein} holds. 

Applying Lemma~\ref{lem:strong-rev-monotone} yields:
\[
\REV(M_{\tilde E}, \Dstar) \geq \REV(M_{\tilde E}, \tilde E) = \OPT(\tilde E) .
\]

Therefore, the remaining task is to ensure that $m$ is sufficiently large such that 
\[
\OPT(\tilde E) \geq (1-\sqrt{\alpha}) \OPT(\Dstar).
\]

We will use a useful lemma below which connects the ratio of revenues that we are interested in with the value of link function at an optimal reserve price.
\begin{lemma}\label{lem:regular-sample-single}
Given two regular distributions $\cD, \bar \cD$ with CDFs $F, \bar F$, such that $\bar F \succeq F$ and $d_k(\cD, \bar \cD) \leq \beta$. Denote the optimal reserve price for $\bar F$ as $\bar P$, and the optimal expected revenue for $F, \bar F$ as $\OPT_F, \OPT_{\bar F}$. Then we have
\[
\frac{\OPT_{F}}{\OPT_{\bar F}}  \geq 1 - \beta h_r(\bar P)
\]
\end{lemma}
\begin{proof}
Recall that $h_r(x) = \frac{1}{1-F(x)}$, and $\bar h_r(x) = \frac{1}{1-\bar F(x)}$. Then, $F(x) \geq \bar F(x)$ implies $h_r(x) \geq \bar h_r(x)$. 

By definition, $d_k(\cD, \bar \cD) \leq \beta$ implies that $\max_x F(x) - \bar F(x) \leq \beta$. So we have:
\begin{align*}
    h_r(x) - \bar h_r(x) = \frac{F(x) - \bar F(x)}{(1-F(x))(1-\bar F(x)} = (F(x) - \bar F(x)) h_r(x) \bar h_r(x) \leq \beta h_r^2(x),
\end{align*}
where the last inequality follows from the fact that $\max_x F(x) - \bar F(x) \leq \beta$, and $h_r(x) \geq \bar h_r(x)$. Thus,  for all $x$,
\begin{align}\label{eq:single-regular-ub-eq1}
    \bar h_R(x) \geq h_r(x) -  \beta h_r^2(x).
\end{align}

Note that the expected revenue, $R(x)= x(1-F(x))$, at any $x$, equals to $\frac{x}{h_r(x)}$, which is the reciprocal of the slope for the linear function $g(a) = h_r(x) \cdot a$. Hence, the revenue is maximized when the slope for the linear function $g(a) = h_r(x) \cdot a$ is minimized. 

Denote the corresponding optimal reserve prices for $F$ and $\bar F$ as $P$ and $\bar P$. Then at $\bar P$, 
\[
\bar h_r(\bar P) = \frac{1}{1-\bar F(\bar P)} = \frac{1}{\OPT_{\bar F}} \cdot \bar P.
\]

Denote $\REV(F, x)$ as the expected revenue with a reserve price at $x$ for a valuation distribution with CDF as $F$. Then,
\[
\frac{\OPT_{F}}{\OPT_{\bar F}}  \geq \frac{\REV(F, \bar P)}{\OPT_{\bar F}} = \frac{\bar h_r(\bar P)}{h(\bar P)} \geq \frac{h_r(\bar P) -  \beta h_r^2(\bar P)}{h_r(\bar P)} = 1 - \beta h_r(\bar P),
\]
where the first inequality follows directly from the definition of the optimal revenue, and the second inequality is from Eq~\eqref{eq:single-regular-ub-eq1}. 
\end{proof}

Now we will use Lemma~\ref{lem:regular-sample-single} to proceed. Denote the optimal reserve price for $\Dstar$ as $P^\ast$. Denote the link function applied to $\tilde E$ and $\Dstar$ as $\tilde h$, $h^\ast$, respectively. Then, we will discuss two cases for $\tilde h(P^\ast)$.   

\textbf{Case 1: $\tilde h(P^\ast) > \frac{1}{\sqrt{\alpha}}$.} 
For this case, $\tilde h(P^\ast) > \frac{1}{\sqrt{\alpha}}$ implies that $q^{\tilde E}( P^\ast) < \sqrt{\alpha}$.
Applying [Lemma 5,~\citet{guo2019settling}] and triangle inequalities, we have
\[
|q^{\tilde E} - q^{\Dstar}| \leq \sqrt{\frac{2 q^{\tilde E}(v) \left( 1-q^{\tilde E}(v) \right) \ln (2m \delta^{-1})}{m}} + \frac{4 \ln (2m \delta^{-1})}{m} + \alpha.
\]
Given that $q^{\tilde E}( P^\ast) < \sqrt{\alpha}$, we have $q^{\tilde E}(1-q^{\tilde E}) \leq q^{\tilde E} \leq \sqrt{\alpha}$. Therefore, it suffices to have
\[
\sqrt{\frac{\sqrt{\alpha}}{m}} \le C_1 \alpha,
\]
for some universal constant $C_1$ to ensure that $|q^{\tilde E} - q^{\Dstar}| = O(\alpha)$, which implies $m \ge \nicefrac{1}{\{C_1^2\alpha^{3/2}\}}$ for some universal constant $C_1$.

\textbf{Case 2: $\tilde h(P^\ast) \leq \frac{1}{\sqrt{\alpha}}$.} 
For this case, $\tilde h(P^\ast) \leq \frac{1}{\sqrt{\alpha}}$ implies that $q^{\tilde E}( P^\ast) \geq \sqrt{\alpha}$.

By lemma~\ref{lem:regular-sample-single}, we have that 
\[
\frac{\OPT_{\tilde E}}{\OPT_{\Dstar}}  \geq 1 - \beta \tilde h_r(P^\ast),
\]
therefore it suffice to ensure that $1 - \beta \tilde h_r(P^\ast) \ge 1 - C_2 \sqrt{\alpha}$ for some universal constant $C_2$, which implies that $\beta \le q^{\tilde E}( P^\ast) \cdot C_2 \sqrt{\alpha}$. Applying [Lemma 5,~\citet{guo2019settling}], it suffices to have that $\sqrt{\frac{q^{\tilde E}( P^\ast)}{m}} \le \beta \leq q^{\tilde E}( P^\ast)\cdot C_2 \sqrt{\alpha}$, which yields that $m > \frac{1}{C_2^2\alpha q^{\tilde E}}$. Lastly, applying the fact that we are in the case where $q^{\tilde E}( P^\ast) \geq \sqrt{\alpha}$ we get that is suffices to have $m > \frac{1}{C_2^2\alpha^{3/2}}$ for some universal constant $C_2$. This completes the proof.

\qed

\subsection{Proof of Theorem \ref{thm:mhr-sample-UB}} \label{app:mhr-sample-UB}

  This follows easily from Theorem \ref{thm:mhr-inf-sample} and the DKW 
inequality \cite{dvoretzky1956asymptotic, massart1990tight} that states that
the empirical CDF with $m$ samples is close to the population CDF with an error
of at most
\[ O\left( \sqrt{\frac{\log(1/\delta)}{m}} \right) \]
with probability at least $1 - \delta$.
\qed

%% file: sec/proofs/sample_LB_proofs.tex
\subsection{Proof of Theorem~\ref{thm:sample-LB}}\label{sec:sample-complexity-lb-proofs}

We omit the details of this proof since it follows from Theorem 2 and
Appendix E of \cite{guo2019settling} applied for the case $n = 1$. The reason
is that if we could get a better bound in our corrupted case then this
algorithm could be used to improve our sample complexity result in the
non-corrupted case.

%% file: robust_auction_arxiv.bbl
\begin{thebibliography}{25}
\providecommand{\natexlab}[1]{#1}
\providecommand{\url}[1]{\texttt{#1}}
\expandafter\ifx\csname urlstyle\endcsname\relax
  \providecommand{\doi}[1]{doi: #1}\else
  \providecommand{\doi}{doi: \begingroup \urlstyle{rm}\Url}\fi

\bibitem[Balcan et~al.(2016)Balcan, Sandholm, and Vitercik]{balcan2016sample}
M.-F. Balcan, T.~Sandholm, and E.~Vitercik.
\newblock Sample complexity of automated mechanism design.
\newblock \emph{arXiv preprint arXiv:1606.04145}, 2016.

\bibitem[Balcan et~al.(2018)Balcan, Sandholm, and Vitercik]{balcan2018general}
M.-F. Balcan, T.~Sandholm, and E.~Vitercik.
\newblock A general theory of sample complexity for multi-item profit
  maximization.
\newblock In \emph{Proceedings of the 2018 ACM Conference on Economics and
  Computation}, pages 173--174, 2018.

\bibitem[Brustle et~al.(2020)Brustle, Cai, and Daskalakis]{brustle2020multi}
J.~Brustle, Y.~Cai, and C.~Daskalakis.
\newblock Multi-item mechanisms without item-independence: Learnability via
  robustness.
\newblock In \emph{Proceedings of the 21st ACM Conference on Economics and
  Computation}, pages 715--761, 2020.

\bibitem[Bykowsky et~al.(2000)Bykowsky, Cull, and
  Ledyard]{bykowsky2000mutually}
M.~M. Bykowsky, R.~J. Cull, and J.~O. Ledyard.
\newblock Mutually destructive bidding: The {FCC} auction design problem.
\newblock \emph{Journal of Regulatory Economics}, 17\penalty0 (3):\penalty0
  205--228, 2000.

\bibitem[Cai and Daskalakis(2011)]{cai2011extreme}
Y.~Cai and C.~Daskalakis.
\newblock Extreme-value theorems for optimal multidimensional pricing.
\newblock In \emph{2011 IEEE 52nd Annual Symposium on Foundations of Computer
  Science}, pages 522--531. IEEE, 2011.

\bibitem[Cai and Daskalakis(2017)]{cai2017learning}
Y.~Cai and C.~Daskalakis.
\newblock Learning multi-item auctions with (or without) samples.
\newblock In \emph{2017 IEEE 58th Annual Symposium on Foundations of Computer
  Science (FOCS)}, pages 516--527. IEEE, 2017.

\bibitem[Cole and Roughgarden(2014)]{cole2014sample}
R.~Cole and T.~Roughgarden.
\newblock The sample complexity of revenue maximization.
\newblock In \emph{Proceedings of the Forty-Sixth Annual ACM Symposium on
  Theory of Computing}, pages 243--252, 2014.

\bibitem[Devanur et~al.(2016)Devanur, Huang, and Psomas]{devanur2016sample}
N.~R. Devanur, Z.~Huang, and C.-A. Psomas.
\newblock The sample complexity of auctions with side information.
\newblock In \emph{Proceedings of the Forty-Eighth Annual ACM Symposium on
  Theory of Computing}, pages 426--439, 2016.

\bibitem[Dud{\'\i}k et~al.(2017)Dud{\'\i}k, Haghtalab, Luo, Schapire,
  Syrgkanis, and Vaughan]{dudik2017oracle}
M.~Dud{\'\i}k, N.~Haghtalab, H.~Luo, R.~E. Schapire, V.~Syrgkanis, and J.~W.
  Vaughan.
\newblock Oracle-efficient online learning and auction design.
\newblock In \emph{2017 IEEE 58th Annual Symposium on Foundations of Computer
  Science}, pages 528--539. IEEE, 2017.

\bibitem[Dvoretzky et~al.(1956)Dvoretzky, Kiefer, and
  Wolfowitz]{dvoretzky1956asymptotic}
A.~Dvoretzky, J.~Kiefer, and J.~Wolfowitz.
\newblock Asymptotic minimax character of the sample distribution function and
  of the classical multinomial estimator.
\newblock \emph{The Annals of Mathematical Statistics}, pages 642--669, 1956.

\bibitem[Gonczarowski and Nisan(2017)]{gonczarowski2017efficient}
Y.~A. Gonczarowski and N.~Nisan.
\newblock Efficient empirical revenue maximization in single-parameter auction
  environments.
\newblock In \emph{Proceedings of the 49th Annual ACM SIGACT Symposium on
  Theory of Computing}, pages 856--868, 2017.

\bibitem[Gonczarowski and Weinberg(2021)]{gonczarowski2021sample}
Y.~A. Gonczarowski and S.~M. Weinberg.
\newblock The sample complexity of up-to-$\varepsilon$ multi-dimensional
  revenue maximization.
\newblock \emph{Journal of the ACM}, 68\penalty0 (3):\penalty0 1--28, 2021.

\bibitem[Guo et~al.(2019)Guo, Huang, and Zhang]{guo2019settling}
C.~Guo, Z.~Huang, and X.~Zhang.
\newblock Settling the sample complexity of single-parameter revenue
  maximization.
\newblock In \emph{Proceedings of the 51st Annual ACM SIGACT Symposium on
  Theory of Computing}, pages 662--673, 2019.

\bibitem[Huang et~al.(2018)Huang, Mansour, and Roughgarden]{huang2018making}
Z.~Huang, Y.~Mansour, and T.~Roughgarden.
\newblock Making the most of your samples.
\newblock \emph{SIAM Journal on Computing}, 47\penalty0 (3):\penalty0 651--674,
  2018.

\bibitem[Klemperer(2002)]{klemperer2002really}
P.~Klemperer.
\newblock What really matters in auction design.
\newblock \emph{Journal of Economic Perspectives}, 16\penalty0 (1):\penalty0
  169--189, 2002.

\bibitem[Lahaie et~al.(2007)Lahaie, Pennock, Saberi, and
  Vohra]{lahaie2007sponsored}
S.~Lahaie, D.~M. Pennock, A.~Saberi, and R.~V. Vohra.
\newblock Sponsored search auctions.
\newblock \emph{Algorithmic Game Theory}, 1:\penalty0 699--716, 2007.

\bibitem[Massart(1990)]{massart1990tight}
P.~Massart.
\newblock The tight constant in the {D}voretzky-{K}iefer-{W}olfowitz
  inequality.
\newblock \emph{The Annals of Probability}, pages 1269--1283, 1990.

\bibitem[Milgrom and Milgrom(2004)]{milgrom2004putting}
P.~Milgrom and P.~R. Milgrom.
\newblock \emph{Putting Auction Theory to Work}.
\newblock Cambridge University Press, 2004.

\bibitem[Morgenstern and Roughgarden(2015)]{morgenstern2015pseudo}
J.~Morgenstern and T.~Roughgarden.
\newblock The pseudo-dimension of near-optimal auctions.
\newblock \emph{arXiv preprint arXiv:1506.03684}, 2015.

\bibitem[Morgenstern and Roughgarden(2016)]{morgenstern2016learning}
J.~Morgenstern and T.~Roughgarden.
\newblock Learning simple auctions.
\newblock In \emph{Conference on Learning Theory}, pages 1298--1318. PMLR,
  2016.

\bibitem[Myerson(1981)]{myerson1981optimal}
R.~B. Myerson.
\newblock Optimal auction design.
\newblock \emph{Mathematics of Operations Research}, 6\penalty0 (1):\penalty0
  58--73, 1981.

\bibitem[Roth and Ockenfels(2002)]{roth2002last}
A.~E. Roth and A.~Ockenfels.
\newblock Last-minute bidding and the rules for ending second-price auctions:
  Evidence from ebay and amazon auctions on the internet.
\newblock \emph{American Economic Review}, 92\penalty0 (4):\penalty0
  1093--1103, 2002.

\bibitem[Roughgarden and Schrijvers(2016)]{roughgarden2016ironing}
T.~Roughgarden and O.~Schrijvers.
\newblock Ironing in the dark.
\newblock In \emph{Proceedings of the 2016 ACM Conference on Economics and
  Computation}, pages 1--18, 2016.

\bibitem[Roughgarden and Wang(2019)]{roughgarden2019minimizing}
T.~Roughgarden and J.~R. Wang.
\newblock Minimizing regret with multiple reserves.
\newblock \emph{ACM Transactions on Economics and Computation}, 7\penalty0
  (3):\penalty0 1--18, 2019.

\bibitem[Syrgkanis(2017)]{syrgkanis2017sample}
V.~Syrgkanis.
\newblock A sample complexity measure with applications to learning optimal
  auctions.
\newblock \emph{arXiv preprint arXiv:1704.02598}, 2017.

\end{thebibliography}
